\documentclass[11pt]{article}

\usepackage[margin=1in]{geometry}
\usepackage{amsmath,amssymb,amsthm,mathtools}
\usepackage{mathrsfs}
\usepackage{bm}
\usepackage{booktabs}
\usepackage{array}
\usepackage{enumitem}
\usepackage[dvipsnames]{xcolor}
\usepackage{tikz}
\usetikzlibrary{positioning,arrows.meta}
\usepackage{microtype}
\usepackage{hyperref}
\usepackage[nameinlink,capitalise,noabbrev]{cleveref}
\usepackage[numbers,sort&compress]{natbib}

\hypersetup{
  colorlinks=true,
  linkcolor=MidnightBlue,
  citecolor=ForestGreen,
  urlcolor=BrickRed
}

\newtheorem{theorem}{Theorem}[section]
\newtheorem{proposition}[theorem]{Proposition}
\newtheorem{corollary}[theorem]{Corollary}
\newtheorem{lemma}[theorem]{Lemma}

\theoremstyle{definition}
\newtheorem{definition}[theorem]{Definition}
\newtheorem{example}[theorem]{Example}
\theoremstyle{remark}
\newtheorem{remark}[theorem]{Remark}

\newcommand{\Ee}{\mathbb E}

\newcommand{\1}{\boldsymbol{1}}

\newcommand{\Pois}{\operatorname{Pois}}

\newcommand{\cM}{\mathcal M}
\newcommand{\ml}[2]{\mathcal{L}(#1\!\to\!#2)}

\newcommand{\dd}{\,\mathrm d}

\newcommand{\resultlink}[2]{\hyperref[#1]{\Cref*{#1}: #2}}

\title{Minimax Quantile Bounds via Information Measures}
\author{Amedeo Roberto Esposito}
\date{\today}

\begin{document}
\maketitle

\begin{abstract}
We develop a unified information-theoretic framework for lower bounding
minimax quantiles. The starting point is a loss-adapted
Neyman--Pearson metaconverse that bounds the minimax success probability
at every loss threshold and confidence level. The bound separates the
small-ball behaviour of the prior under the loss from the statistical
distinguishability of the observation model, and is optimised over an auxiliary
output distribution. Different relaxations of this Neyman--Pearson bound
yield converses based on \(f\)-informativity, Sibson mutual information
\(I_\alpha\), Maximal Leakage, and Amemiya norms. Classical Fano and
Le Cam lower bounds are recovered as special cases.

The framework also clarifies why different information measures are
suited to different recovery criteria. Maximal Leakage is exact for a
class of symmetric exact-recovery problems. We use this identity to
derive finite-sample bounds on the full minimax exact-recovery risk in
the balanced Gaussian weighted stochastic block model, as well as
two-sided finite-sample minimax-quantile bounds for low-rank matrix
estimation under isotropic bounded-energy noise. For approximate Hamming
recovery, we exhibit a heterogeneous binary model in which an optimised
finite Sibson order yields a strong converse while the Maximal Leakage
specialisation is trivial. Finally, for one-coordinate Poisson
localisation, a Bennett-type Young function used through its Amemiya norm
recovers the exact success-probability scale, whereas classical Fano and
fixed-power relaxations are strictly weaker. These results show that sharp
converses for minimax quantiles require adapting the information measure to
the recovery resolution, whether exact or approximate, and to the tail
behaviour of the likelihood ratio.
\end{abstract}

\section{Introduction}
Suppose that we wish to estimate an unknown parameter
\(\theta\in\Theta\) from noisy observations \(X^n\) taking values in a
measurable space \(\mathcal X^n\). The law of \(X^n\) depends on
\(\theta\), and the quality of an estimator \(\widehat\theta(X^n)\) is
measured by a loss \(\ell(\widehat\theta,\theta)\). Classical minimax
theory usually summarises this random loss through its expectation.
Here we instead ask a confidence-dependent question: for a prescribed
\(\delta\in(0,1)\), what is the smallest loss threshold for which the
worst-case failure probability is at most \(\delta\)? This is the minimax
quantile problem systematically formulated by Ma et al.
\cite{MaVerchandSamworth2024}. It retains the tail of the loss and can
therefore distinguish observation models that have the same expected-risk rate
but very different probabilities of a large estimation error.

This paper develops a unified approach to converse bounds for minimax
quantiles. The organising object is the Neyman--Pearson function
\(\beta_\alpha(P,Q)\). It plays the role of an information measure by
quantifying the distinguishability between \(P\) and \(Q\). Introduce an auxiliary prior \(P_W\) on
\(\Theta\), let \(X\) be generated conditionally on \(W\), and consider
the success event
\(A_\rho=\{\ell(\widehat\theta(X),W)<\rho\}\). Under the true joint distribution
\(P_{WX}\), the probability of \(A_\rho\) is the average probability of
achieving loss smaller than \(\rho\). Define the small-ball function
as the largest probability that the prior assigns to a loss ball:
\[
L_W(\rho):=
\sup_{\vartheta\in\Theta}
P_W\{\ell(\vartheta,W)<\rho\}.
\]
Under the product reference distribution \(P_WQ_X\), the parameter \(W\) and
the observation \(X\) are independent. Under independence, the probability
of success cannot exceed the small-ball probability, so
\(P_WQ_X(A_\rho)\leq L_W(\rho)\) for every estimator. By the definition of
the inverse Neyman--Pearson function, the probability of the same event under
the true joint distribution satisfies
\[
P_{WX}(A_\rho)
\leq
\beta_{\circ}(P_{WX},P_WQ_X)^{-1}\bigl(L_W(\rho)\bigr).
\]
This inequality holds for every auxiliary distribution \(Q_X\), so we
minimise the right-hand side over \(Q_X\). If the resulting upper bound is
smaller than \(1-\delta\), no estimator can achieve loss below \(\rho\) with
probability \(1-\delta\). Hence the minimax quantile is at least \(\rho\); see
\Cref{thm:npLowerBound} in \Cref{sec:main-results}.

The two quantities play different roles. The small-ball function
\(L_W(\rho)\) depends only on the prior and the loss. The Neyman--Pearson
function depends on the observation model. This allows the exact
Neyman--Pearson bound to be replaced by several tractable inequalities.
Binary data processing for an
\(f\)-divergence gives an \(f\)-informativity converse. A uniform prior on
a finite packing then recovers the \(f\)-Fano method; choosing $f(x)=x\log x$ gives
classical Fano, and the two-point total-variation choice gives Le Cam's
method. Binary R\'enyi data processing gives a converse in terms of
Sibson mutual information \(I_\alpha\), and the limit \(\alpha\to\infty\)
leads to Maximal Leakage. Alternatively, generalised H\"older duality in
Orlicz spaces controls the same success event through the Amemiya norm of
the likelihood ratio. These specialisations are stated in
\Cref{thm:fDivLowerBound,cor:renyiSibson,cor:maximalLeakageQuantile,cor:amemiyaHolder}.
Thus, the Neyman--Pearson bound is the common starting point. The
converses above follow by applying different inequalities to it.
Classical Fano and Le Cam are recovered in
\Cref{rmk:classicalFano,cor:leCam}, and the full hierarchy is summarised
in \Cref{fig:lower-bound-hierarchy}.

No single relaxation is uniformly preferable. Because Maximal Leakage is
based on the pointwise supremum of the conditional densities, it is
particularly well matched to exact identification; see
\Cref{sec:exact-recovery-maximal-leakage,thm:maximalLeakageExact}. When approximate recovery is allowed, success includes a neighbourhood of the
true parameter rather than requiring exact identification. A finite
\(I_\alpha\) can exploit how probability is distributed across this
larger success event, whereas Maximal Leakage fails to do so, as shown in
\Cref{sec:bsc-hamming,prop:finiteSibsonBeatsLeakage}. When rare observations
determine the success probability, fixed power moments can obscure the
correct scale, and a suitably chosen Young function can be sharper; see
\Cref{sec:amemiya-outlier,prop:taggedPoisson}.

\subsection{Contributions.}
The main contributions are the following
\begin{enumerate}
\item We establish a loss-adapted Neyman--Pearson metaconverse that bounds the
minimax success probability at every loss threshold. Inverting this bound
yields minimax-quantile lower bounds for every confidence level
(\Cref{thm:npLowerBound}). The auxiliary prior may be chosen freely, and
the bound is tightened by optimising over the reference distribution
\(Q_X\); see \Cref{sec:main-results}.

\item We derive a common hierarchy of tractable specialisations:
binary \(f\)-divergence inversion and \(f\)-informativity, binary
R\'enyi inversion and Sibson mutual information, Maximal Leakage, and
Amemiya--H\"older bounds. Classical Fano and Le Cam appear as corollaries
of the same construction. These classical inequalities are not presented
as new; their recovery identifies precisely where the usual methods sit
inside the metaconverse. The corresponding results are
\Cref{thm:fDivLowerBound,cor:renyiSibson,cor:maximalLeakageQuantile,cor:amemiyaHolder},
with Fano and Le Cam recovered in
\Cref{cor:fFano,rmk:classicalFano,cor:leCam}; see also
\Cref{fig:lower-bound-hierarchy}.

\item Under a uniform prior on a finite parameter set, we show that the
Bayes success probability equals the normalised exponential of Maximal
Leakage. If the MAP rule is an equaliser, the same identity gives the
exact minimax success probability. The equaliser assumption is essential:
Maximal Leakage need not be exact for arbitrary recovery problems; see
\Cref{thm:maximalLeakageExact} in
\Cref{sec:exact-recovery-maximal-leakage}.

\item As a first application, we derive, to the best of our knowledge,
the first finite-sample bounds on the full minimax exact-recovery risk in
the balanced two-community Gaussian weighted stochastic block model.
Unlike the known first-order threshold \cite{PandeyKulkarni2024}, our
bounds retain the dependence on the network size and signal strength, and
reduce the recovery curve to explicit Gaussian extreme-value
probabilities.
The result is stated in \Cref{thm:gwsbmFiniteLeakage}, with its
strong-converse consequence in \Cref{cor:sbmStrongConverse}; see
\Cref{sec:gaussian-weighted-sbm}.

As a second application, we give, to the best of our knowledge, the first
two-sided finite-sample minimax-quantile bounds for continuous low-rank
matrix estimation under isotropic uniform Frobenius-ball noise. These
bounds identify, up to universal constants, the joint dependence on the
intrinsic low-rank dimension and the confidence level; see
\Cref{prop:boundedEnergyLowRankQuantile} in
\Cref{sec:boundedEnergyLowRank}.

\item We prove two strict separation results. For approximate recovery
over a heterogeneous binary channel, an optimised finite Sibson order
gives an exponentially vanishing upper bound on minimax success, and
hence a strong converse, while the Maximal Leakage specialisation is
trivial; see
\Cref{prop:heterogeneousSibsonConverse,prop:finiteSibsonBeatsLeakage,ex:heterogeneousSibsonSeparation}
in \Cref{sec:bsc-hamming}. For one-coordinate Poisson localisation, the Bennett-type Young
function used through its Amemiya norm recovers the exact
\(\log M/(M\log\log M)\) success scale. Classical Fano proves only that
success vanishes, at the much weaker \(1/\log M\) scale. The latter
example shows that leveraging generalised norms is useful not as a universal
replacement for information divergences, but when the success probability
is determined by rare observations that fixed-order moments do not capture;
see \Cref{prop:taggedPoisson} in
\Cref{sec:amemiya-outlier}.
\end{enumerate}

\subsection{Related work}

\paragraph{Classical minimax lower bounds and minimax quantiles.}
Fano's inequality and Le Cam's two-point method are foundational
reductions from estimation to testing
\cite{Fano1961,LeCam1973,Yu1997,Wainwright2019}. Their extensions include
general \(f\)-divergence bounds, continuum and distance-based Fano
inequalities, and formulations for general random variables
\cite{Gushchin2003,Guntuboyina2011,DuchiWainwright2013,
GerchinovitzMenardStoltz2020}. Most of this literature targets expected
minimax risk or a fixed error probability. Ma et al.
\cite{MaVerchandSamworth2024} introduced minimax quantiles as a general
high-probability criterion and developed quantile versions of Fano and
Le Cam, together with local-to-global reductions. We use the same
quantile objective, but start from a single optimised binary-testing
metaconverse. Their high-probability Fano and Le Cam bounds are recovered
as particular relaxations, while our framework also accommodates
\(f\)-informativity, Sibson mutual information, Maximal Leakage and Amemiya norms.

\paragraph{Testing metaconverses.}
The use of an optimal binary test begins with the
Neyman--Pearson lemma \cite{NeymanPearson1933}. Polyanskiy et al.
\cite{PolyanskiyPoorVerdu2010} used this idea to derive the meta-converse
for finite-blocklength channel coding. They compare the joint distribution
of the transmitted message and the channel output with a reference
distribution under which these variables are independent. This comparison
bounds the probability of correct decoding, and optimising over the
auxiliary output distribution tightens the bound.
Venkataramanan et al.
\cite{VenkataramananJohnson2018} adapted this perspective to multiple
hypothesis testing and high-dimensional estimation, showing how it can
sharpen Fano and yield strong converses. Xu et al.
\cite{XuRaginsky2017} connected the \(\beta\)-function, prior small-ball
probabilities and Bayes risk; they also proved a conditional version with
an auxiliary random variable. Relations among the Neyman--Pearson function,
divergences and Bayes risk are developed more generally
in \cite{ReidWilliamson2011}. Our approach instead begins with the event
that the loss is below a certain threshold. We optimise the inverse
Neyman--Pearson bound over \(Q_X\), invert the result to obtain a lower
bound on the minimax quantile, and then derive bounds involving different
information measures and Amemiya norms.

\paragraph{\(f\)-divergences and informativity.}
The general theory of \(f\)-divergences originates with
\cite{AliSilvey1966,Csiszar1967}, and \(f\)-informativity was introduced
in \cite{Csiszar1972}. Binary \(f\)-divergence forms of Fano-type minimax
bounds were developed in \cite{Gushchin2003,Guntuboyina2011}. Chen et al.
\cite{ChenGuntuboyinaZhang2016} gave a general \(f\)-informativity
framework for Bayes-risk lower bounds under arbitrary priors and losses.
The binary inverse used here is consequently a known ingredient, and our
\(f\)-Fano corollary agrees with these results after translating the
success event into a zero--one loss. Our starting point is instead the
optimised Neyman--Pearson bound, which admits specialisations beyond
\(f\)-informativity.

More broadly, Esposito et al.
\cite{EspositoVandenbroucqueGastpar2024} used dual representations and
data processing to obtain Bayes-risk bounds from a wide range of
information measures, including \(f\)-divergences, R\'enyi divergences
and Sibson mutual information. The present work differs in its
minimax-quantile objective, its direct optimisation of the full
Neyman--Pearson success curve, and its emphasis on sharp separation
examples showing that the preferred relaxation changes between exact and
approximate recovery and according to whether rare observations determine
the success probability.

\paragraph{R\'enyi divergence, Sibson mutual information and Maximal Leakage.}
R\'enyi divergence and Sibson mutual information were introduced in
\cite{Renyi1961,Sibson1969}; modern accounts of their properties and
variational representations include
\cite{vanErvenHarremoes2014,Verdu2015,EspositoGastparIssa2025}. Maximal
Leakage quantifies the largest multiplicative increase in the probability
of correctly guessing any randomised function of a ``secret'' random
variable \cite{IssaWagnerKamath2020}. This
interpretation explains its natural
role in exact recovery. Our contribution is not the
definition of these measures, but their placement in a single quantile
metaconverse and the resulting comparison across recovery criteria:
Maximal Leakage is exact for symmetric exact recovery problems, whereas
a finite Sibson order can yield a strong converse for approximate recovery
when Maximal Leakage gives no nontrivial exponent.

\paragraph{Interactive and other unified lower-bound frameworks.}
Recent work has unified Assouad, Fano and Le Cam reductions in
interactive decision problems \cite{ChenFosterHanQianRakhlinXu2024}, and
high-probability Fano and Le Cam methods have been extended to
interactive and privacy-constrained minimax quantiles
\cite{BongoleOechteringSkoglund2026,
BongoleZamaniOechteringSkoglund2026}. Those results address interaction,
learnability or privacy constraints. The present paper instead studies a
general non-interactive estimation setting and enlarges the range of bounds
available after the testing reduction, from \(f\)-informativity to
finite-order Sibson mutual information, Maximal Leakage and Amemiya norms chosen
to reflect the likelihood-ratio distribution.
\section{Preliminaries}
\subsection{Minimax quantile risk}
Let $(\Theta, d)$ be a pseudo-metric space \textit{i.e.}, $d:\Theta \times \Theta \to [0,\infty) $ is a function such that for every $\theta,\zeta,\omega \in \Theta$ one has that $d(\theta,\theta)=0$, $d(\theta,\zeta)=d(\zeta,\theta)$ and $d(\theta,\zeta)\leq d(\theta,\omega)+d(\omega,\theta)$. However, differently from a metric space, one may have $d(\theta,\zeta)=0$ even if $\theta\neq \zeta$. For each \(\theta\in\Theta\), let \(\mathscr P_\theta\) be a nonempty collection of probability measures on \((\mathcal X,\mathcal A)\). Under i.i.d. sampling, define the corresponding \(n\)-sample class by
\[
\mathscr P_\theta^{(n)}
:=
\{P^{\otimes n}:P\in\mathscr P_\theta\}.
\]
Given \(\theta\), we observe \(X^n\sim P^{\otimes n}\) for some \(P\in\mathscr P_\theta\).
Let $g:[0,\infty) \to [0,\infty)$ be a non-decreasing function and define a loss function $\ell:\Theta \times\Theta \to [0,\infty)$ as follows
$$ \ell(\theta,\zeta)= g(d(\theta,\zeta)),$$
for $\theta,\zeta \in \Theta$.
Let
\[
\mathcal E_n
:=
\left\{
\widehat\theta:\mathcal X^n\to\Theta:
\widehat\theta\text{ is measurable}
\right\}
\]
denote the class of estimators. For \(\widehat\theta\in\mathcal E_n\), the random variable \(\widehat\theta(X^n)\) based on the observation \(X^n\). Under parameter \(\theta\), the observation has distribution\[X^n\sim P_{X^n\mid\theta}.\]
For \(\delta\in(0,1]\), \(\theta\in\Theta\), \(\widehat\theta\in\mathcal E_n\), and \(P_{X^n\mid\theta}\in\mathscr P_{X^n\mid\theta}\), define the $(1-\delta)$-th quantile as follows
\[
Q_{1-\delta}
\bigl(\widehat\theta,\theta;P_{X^n\mid\theta}\bigr)
:=
\inf\left\{
r\ge0:
P_{X^n\mid\theta}
\left(
\ell\bigl(\widehat\theta(X^n),\theta\bigr)\le r
\right)
\ge1-\delta
\right\}.
\]

When the observation law is omitted, \(Q_{1-\delta}\) denotes the corresponding worst-case quantile:
\[
Q_{1-\delta}(\widehat\theta,\theta)
:=
\sup_{P_{X^n\mid\theta}\in\mathscr P_{X^n\mid\theta}}
Q_{1-\delta}
\bigl(\widehat\theta,\theta;P_{X^n\mid\theta}\bigr).
\]
When \(\theta\) determines the observation law uniquely, \(\mathscr P_{X^n\mid\theta}\) is a singleton, and the two notions of quantile coincide.
\begin{definition}[Minimax quantile~\cite{MaVerchandSamworth2024}]
\label{def:minimax-quantile}
Let \(\delta \in (0,1]\) the minimax \((1-\delta)\)-quantile of the loss $\ell$ for estimating \(\theta\) via \(X^n\) can be defined as follows

\[
\mathcal M_n(\delta)
:=
\inf_{\widehat\theta\in\mathcal E_n}
\sup_{\theta\in\Theta}
Q_{1-\delta}(\widehat\theta,\theta).
\]
\end{definition}
A related quantity, which is more amenable to analysis and easier to interpret operationally, is the so-called lower minimax quantile.
\begin{definition}[Lower minimax quantile~\cite{MaVerchandSamworth2024}]
\label{def:lower-minimax-quantile}
Let \(\delta \in (0,1]\), the lower minimax \((1-\delta)\)-quantile is defined as follows
\[
\mathcal M_{n,-}(\delta) := \inf\left\{ r\geq 0 : 
\inf_{\widehat\theta\in\mathcal E_n}
\sup_{\theta\in\Theta}\sup_{P_{X^n\mid\theta}\in\mathscr P_{X^n\mid\theta}}P_{X^n\mid\theta}
\left(
\ell\bigl(\widehat\theta(X^n),\theta\bigr)> r
\right)
\le \delta
\right\}.
\]
The lower minimax quantile is the smallest loss threshold for which one can choose an estimator whose worst-case failure probability is at most \(\delta\).
Moreover, since one has that $\mathcal M_n(\delta) \geq \mathcal M_{n,-}(\delta)$ (\cite[Theorem 4]{MaVerchandSamworth2024}) providing a lower-bound on the lower minimax quantile yields a lower-bound on the minimax quantile itself. Moreover, if one can prove that for some $r$
\[
\inf_{\widehat\theta\in\mathcal E_n}
\sup_{\theta\in\Theta}\sup_{P_{X^n\mid\theta}\in\mathscr P_{X^n\mid\theta}}P_{X^n\mid\theta}
\left(
\ell\bigl(\widehat\theta(X^n),\theta\bigr)> r
\right)
> \delta
\]
then one has that $\mathcal M_n(\delta)\geq \mathcal M_{n,-}(\delta) \geq r$.
\end{definition}
For a sequence of estimation problems, we call a converse \emph{strong}
if, whenever its stated impossibility condition holds, its lower bound on
the minimax failure probability converges to one~\cite{VenkataramananJohnson2018}.  
A fundamental object that appears throughout the manuscript is the so-called small-ball probability~\cite{XuRaginsky2017,ChenGuntuboyinaZhang2016,EspositoVandenbroucqueGastpar2024}.
\begin{definition}[Small-ball function]
For an auxiliary prior $P_W$ on $\Theta$ and a threshold $\rho>0$, define
\begin{equation}
L_W(\rho)
:=
\sup_{\vartheta\in\Theta}
P_W\bigl(\ell(\vartheta,W)<\rho\bigr).
\label{eq:smallBallFunction}
\end{equation}
\end{definition}
\subsection{Information measures}
We now introduce the various families of information measures we will use to lower bound the risk.
\begin{definition}[$f$-divergence and $f$-information]
Let $f:(0,\infty)\to\mathbb R$ be convex with $f(1)=0$.  When $P\ll Q$,
\begin{equation}
D_f(P\|Q)
:=
\int f\left(\frac{\mathrm dP}{\mathrm dQ}\right)\mathrm dQ,
\end{equation}
with the standard lower-semicontinuous extension when $P$ has a singular
part~\cite{AliSilvey1966,Csiszar1967}.  For $p,q\in[0,1]$, its binary form is
\begin{equation}
d_f(p\|q)
:=
qf\left(\frac pq\right)
+(1-q)f\left(\frac{1-p}{1-q}\right),
\end{equation}
For fixed $q$, we invert $p\mapsto d_f(p\|q)$ over $p\in[q,1]$ by defining
\begin{equation}
d_f^{-1}(a\|q)
:=
\sup\{p\in[q,1]:d_f(p\|q)\leq a\}.
\end{equation}
See~\cite[Eq.~(5) and Proposition~1]{MajumdarMeiPacelli2023} for this inverse and its monotonicity on $[q,1]$.
The $f$-informativity is instead defined as follows~\cite{Csiszar1972,ChenGuntuboyinaZhang2016}
\begin{equation}
I_f(W,X)
:=
\inf_{Q_X}D_f(P_{WX}\|P_WQ_X).
\end{equation}
For a general $f$-divergence, the optimisation over $Q_X$ need not admit
a closed-form solution, which can make $f$-informativity less amenable to
direct analysis and applications.
\end{definition}
\begin{definition}[Sibson mutual information and Maximal Leakage]
For $\alpha>1$, the R\'enyi divergence and Sibson mutual information of order
$\alpha$ are, respectively~\cite{Renyi1961,vanErvenHarremoes2014,Sibson1969,Verdu2015},
\begin{align}
D_\alpha(P\|Q)
&:=
\frac1{\alpha-1}
\log\int
\left(\frac{\mathrm dP}{\mathrm dQ}\right)^\alpha\mathrm dQ,
\\
I_\alpha(W,X)
&:=
\inf_{Q_X}D_\alpha(P_{WX}\|P_WQ_X).
\end{align}
For $p,q\in[0,1]$, define the binary R\'enyi divergence and, for fixed
$q$, its inverse over $p\in[q,1]$ by
\begin{align}
d_\alpha(p\|q)
&:=
\frac1{\alpha-1}
\log\left\{
p^\alpha q^{1-\alpha}
+(1-p)^\alpha(1-q)^{1-\alpha}
\right\},
\\
d_\alpha^{-1}(a\|q)
&:=
\sup\{p\in[q,1]:d_\alpha(p\|q)\leq a\}.
\end{align}
Maximal Leakage is the limit~\cite{IssaWagnerKamath2020}
\begin{equation}
\ml{W}{X}
:=
\lim_{\alpha\to\infty}I_\alpha(W,X).
\label{eq:maximalLeakageDefinition}
\end{equation}
If the channel is dominated by a measure $\mu$ and has conditional
densities $p_{X\mid W=w}$, this limit has the following closed-form expression~\cite{IssaWagnerKamath2020}
\begin{equation}
\ml{W}{X}
=
\log\int
\operatorname*{ess\,sup}_{w\in\operatorname{supp}(P_W)}
p_{X\mid W=w}(x)\,\mu(\mathrm dx).
\label{eq:maximalLeakageIntegralDefinition}
\end{equation}
\end{definition}
Young functions provide a third family of bounds, distinct from
\(f\)-informativity and Sibson mutual information. We first define complementary
Young functions and the two norms used below.
\begin{definition}[Complementary Young functions and Amemiya norms~\cite{RaoRen1991,HudzikMaligranda2000}]
A Young function is a convex, non-decreasing function
$\Psi:[0,\infty)\to[0,\infty]$ with $\Psi(0)=0$ and
$\Psi(t)\to\infty$ as $t\to\infty$. Its complementary Young function is
\begin{equation}
\Psi^\star(s):=\sup_{t\geq0}\{st-\Psi(t)\}.
\end{equation}
For a probability measure $Q$, write
\begin{equation}
\|Y\|_{L_\Psi^\mathrm L(Q)}
:=
\inf\left\{c>0:\mathbb E_Q\Psi\left(\frac{|Y|}{c}\right)\leq1\right\}
\end{equation}
for the Luxemburg norm and
\begin{equation}
\|Z\|_{L_{\Psi^\star}^\mathrm A(Q)}
:=
\inf_{t>0}
\frac{1+\mathbb E_Q\Psi^\star(t|Z|)}{t}
=
\sup\left\{\mathbb E_Q|ZY|:\|Y\|_{L_\Psi^\mathrm L(Q)}\leq1\right\}
\end{equation}
for the Amemiya norm.  The equality between the infimum and supremum
representations holds without additional assumptions
\cite{HudzikMaligranda2000}.
The supremum representation shows that the Amemiya norm is the norm dual
to the Luxemburg norm. This dual norm is often called the Orlicz norm;
throughout the remainder of the paper, we use only the term
\emph{Amemiya norm}.
Consequently, one has the generalised H\"older-Rogers' inequality~\cite{RaoRen1991,HudzikMaligranda2000}
\[
\mathbb E_Q|ZY|
\leq
\|Z\|_{L_{\Psi^\star}^\mathrm A(Q)}
\|Y\|_{L_\Psi^\mathrm L(Q)}.
\]
\end{definition}
Going beyond the classical \(L_\alpha\) norms also allows us to treat settings involving heavy-tailed random variables~\cite{EspositoMondelli}. In such settings, the \(L_\alpha\) norms may be infinite even when the KL divergence is finite, and the resulting Fano-type lower bound may fail to recover the correct minimax rate.
\section{Main Results}
\label{sec:main-results}
We now provide a general Information-Theoretic result that allows us to lower-bound the minimax quantile in a variety of problems, with a few corresponding relaxations that can be useful in certain applications. The lower-bound argument introduces an auxiliary probability measure \(\pi\) on the parameter space \(\Theta\). Let \(W\sim\pi\), and, conditional on \(W=\theta\), let \(X^n\sim P_{X^n\mid\theta}\) for some \(P_{X^n\mid\theta}\in\mathscr P_{X^n\mid\theta}\). Since the worst-case failure probability is larger than its average under any such auxiliary distribution,
\begin{equation}
\inf_{\widehat\theta\in\mathcal E_n}
\sup_{\theta\in\Theta}
\sup_{P_{X^n\mid\theta}\in\mathscr P_{X^n\mid\theta}}
P_{X^n\mid\theta}
\{\ell(\widehat\theta(X^n),\theta)\ge\rho\}
\ge
\inf_{\widehat\theta\in\mathcal E_n}
P_{WX^n}
\{\ell(\widehat\theta(X^n),W)\ge\rho\}.\label{eq:BayesianReduction}
\end{equation}
Consequently, lower-bounding the right-hand side of~\Cref{eq:BayesianReduction} leads to a lower-bound on the quantile through the argument described above. In order to do so, let us introduce the Neyman-Pearson function~\cite{NeymanPearson1933,PolyanskiyPoorVerdu2010}
\begin{equation}
    \beta_\alpha(P,Q) := \inf\{ \mathbb{E}_Q[\varphi]: 0\leq \varphi \leq 1, \mathbb E_P[\varphi]\geq \alpha \}.
\end{equation}
The Neyman--Pearson (NP) function was connected to Bayesian risk by
Xu et al.~\cite{XuRaginsky2017}. Let us also define its generalised inverse.
\begin{align}
\beta_{\circ}(P,Q)^{-1}(u)
&:=
\sup\{\alpha\in[0,1]:\beta_\alpha(P,Q)\leq u\}
\\
&=
\sup\{\mathbb E_P[\varphi]:0\leq\varphi\leq1,
\ \mathbb E_Q[\varphi]\leq u\}.
\end{align}
The definition of $\beta_\alpha$ above is identical to~\cite[Eq.~(100)]{PolyanskiyPoorVerdu2010}. For the event $A_\rho=\{\ell(\widehat\theta(X),W)<\rho\}$, Appendix~A of~\cite{XuRaginsky2017} proves
\[
\beta_{P_{WX}(A_\rho)}(P_{WX},P_WP_X)\leq L_W(\rho).
\]
The theorem below proves the corresponding inequality for every $Q_X$, takes the infimum over $Q_X$, and then applies~\eqref{eq:BayesianReduction}. Xu et al. also prove a conditional version involving an auxiliary random variable, which is not considered here.
We can now state the first result, connecting the inverse NP-function to the minimax quantile.
\begin{theorem}\label{thm:npLowerBound}
    For every estimator $\widehat \theta$, auxiliary prior $P_W$, channel $P_{X|W}$ one has the following lower-bound 
    \begin{equation}
          P_{WX}(\ell(\widehat \theta(X),W)\geq \rho) \geq 1-\inf_{Q_X}\beta_{\circ}(P_{WX},P_WQ_X)^{-1}(L_W(\rho)).
    \end{equation}
    Consequently, let \(\delta \in (0,1)\), if \[\inf_{Q_X}\beta_{\circ}(P_{WX},P_WQ_X)^{-1}(L_W(\rho)) <1-\delta \] then
    \begin{equation}
        \mathcal M_{-}(\delta) \geq \rho \text{ and } \mathcal M(\delta) \geq \rho.
    \end{equation}
\end{theorem}
One can then obtain, as a corollary, bounds involving arbitrary
$f$-divergences and generalised Amemiya norms.
\begin{corollary}\label{thm:fDivLowerBound}
Under the same assumptions as~\Cref{thm:npLowerBound} one has the following lower-bound
\begin{equation}
      P_{WX}(\ell(\widehat \theta(X),W)\geq \rho) \geq 1-\inf_{Q_X}d_{f}^{-1}(D_f(P_{WX}\|P_WQ_X)\|L_W(\rho)).
\end{equation}
  Consequently, let \(\delta \in (0,1)\), if \[\inf_{Q_X}d_{f}^{-1}(D_f(P_{WX}\|P_WQ_X)\|L_W(\rho)) <1-\delta \] then
    \begin{equation}
        \mathcal M_{-}(\delta) \geq \rho \text{ and } \mathcal M(\delta) \geq \rho.
    \end{equation}
\end{corollary}
\begin{remark}
    One can translate~\Cref{thm:fDivLowerBound} into a statement involving
    $f$-informativity.  Suppose that $1-\delta\geq L_W(\rho)$ and
    \[ I_f(W,X) < d_f(1-\delta\|L_W(\rho)),\] by definition of infimum there exists a $\tilde Q_X$ such that
    \[
    D_f(P_{WX}\|P_W \tilde{Q}_X) < d_f(1-\delta\|L_W(\rho)).
    \]
    By monotonicity of the binary divergence on the branch
    $[L_W(\rho),1]$, this implies that
    \begin{equation}
        \inf_{Q_X}d_{f}^{-1}(D_f(P_{WX}\|P_WQ_X)\|L_W(\rho)) <1-\delta.
    \end{equation}
    By~\Cref{thm:fDivLowerBound} this implies that
    \[
       \mathcal M_{-}(\delta) \geq \rho \text{ and } \mathcal M(\delta) \geq \rho.
    \]
    The same reasoning applies to \(f\)-mutual information by taking
    \(Q_X=P_X\), where \(P_X\) is the marginal induced by \(P_{WX}\).
\end{remark}
To obtain a prior-free minimax lower bound, take the prior to be uniform
on a finite separated set. Its cardinality controls the small-ball
probability and gives the following \(f\)-Fano inequality.

\begin{corollary}
\label{cor:fFano}
Under the same assumptions as~\Cref{thm:fDivLowerBound}, let
\[
\{\theta_1,\ldots,\theta_M\}\subseteq\Theta
\]
be a finite packing such that, for every $\vartheta\in\Theta$,
\begin{equation}
    \left|
    \left\{
    j\in[M]:
    \ell(\vartheta,\theta_j)<\rho
    \right\}
    \right|
    \leq 1.
\end{equation}
Let $P_W$ be the uniform distribution over
$\{\theta_1,\ldots,\theta_M\}$. Then
\begin{equation}
    L_W(\rho)\leq \frac{1}{M}.
\end{equation}
Consequently, one has
\begin{equation}
    P_{WX}\bigl(\ell(\widehat\theta(X),W)\geq\rho\bigr)
    \geq
    1-
    \inf_{Q_X}
    d_f^{-1}
    \left(
    D_f(P_{WX}\|P_WQ_X)
    \middle\|
    \frac{1}{M}
    \right).
\end{equation}
In particular, let $\delta\in(0,1)$ with $\delta\leq1-1/M$ and suppose that
\begin{equation}
   I_f(W,X)
    <
    d_f\left(
    1-\delta
    \middle\|
    \frac{1}{M}
    \right).
\end{equation}
Then
\begin{equation}
    \mathcal M_{-}(\delta)\geq\rho
    \qquad\text{and}\qquad
    \mathcal M(\delta)\geq\rho.
\end{equation}
\end{corollary}
This result corresponds to a generalisation of Fano's method to arbitrary $f$-divergences. For a uniform prior on a finite packing, the resulting $f$-Fano bound is equivalent to earlier $f$-divergence bounds~\cite{Gushchin2003,Guntuboyina2011}. Applying~\cite[Theorem~2]{ChenGuntuboyinaZhang2016} to the loss $\mathbf 1\{\ell\geq\rho\}$ gives the same $f$-informativity condition after taking complements. The upper inverse $d_f^{-1}$ also appears in~\cite[Theorem~1]{MajumdarMeiPacelli2023}. The results in~\cite{ChenFosterHanQianRakhlinXu2024,BongoleOechteringSkoglund2026,BongoleZamaniOechteringSkoglund2026} include analogous statements for interactive problems or bounded transformations of the loss. Therefore,~\Cref{cor:fFano} is not a new $f$-divergence form of Fano's inequality; here it is obtained from~\Cref{thm:npLowerBound} and is stated as a sufficient condition for the minimax-quantile lower bound.
In particular, with a specific choice of $f$ one recovers Fano's inequality.
\begin{remark} \label{rmk:classicalFano}
Consider the Kullback--Leibler divergence~\cite{KullbackLeibler1951},
corresponding to $f(t)=t\log t$.
Let $W$ be uniformly distributed over the packing
$\{\theta_1,\ldots,\theta_M\}$ and let
$P_{X\mid W=\theta_j}=P_j$. Then
\begin{equation}
    \inf_{Q_X}
    D_{\mathrm{KL}}(P_{WX}\|P_WQ_X)
    =
    \frac1M
    \inf_{Q_X}
    \sum_{j=1}^M
    D_{\mathrm{KL}}(P_j\|Q_X)
    =
    I(W;X).
\end{equation}
For the Kullback--Leibler divergence, the minimiser is the marginal
$P_X=\frac1M\sum_{j=1}^M P_j$. Therefore, the infimum on the left is
$D_{\mathrm{KL}}(P_{WX}\|P_WP_X)$, which is precisely the mutual
information $I(W;X)$.

By~\Cref{cor:fFano}, the following condition
\begin{equation}
    I(W;X)
    <
    d_{\mathrm{KL}}
    \left(
    1-\delta
    \middle\|
    \frac1M
    \right)
\end{equation}
implies
\begin{equation}
    \mathcal M_{-}(\delta)\geq\rho
    \qquad\text{and}\qquad
    \mathcal M(\delta)\geq\rho.
\end{equation}
Since
\begin{align}
    d_{\mathrm{KL}}
    \left(
    1-\delta
    \middle\|
    \frac1M
    \right)
    &=
    (1-\delta)\log M
    -
    h_2(\delta)
    -
    \delta\log\left(1-\frac1M\right),
\end{align}
and
\begin{equation}
    \sup_{\delta\in[0,1]}
    \left\{
    h_2(\delta)
    +
    \delta\log\left(1-\frac1M\right)
    \right\}
    =
    \log\left(2-\frac1M\right),
\end{equation}
one has
\begin{equation}
    d_{\mathrm{KL}}
    \left(
    1-\delta
    \middle\|
    \frac1M
    \right)
    \geq
    (1-\delta)\log M
    -
    \log\left(2-\frac1M\right).
\end{equation}
Therefore, the sufficient condition
\begin{equation}
    \frac{
    I(W;X)+\log\left(2-\frac1M\right)
    }{
    \log M
    }
    <
    1-\delta
\end{equation}
implies
\begin{equation}
    \mathcal M_{-}(\delta)\geq\rho
    \qquad\text{and}\qquad
    \mathcal M(\delta)\geq\rho.
\end{equation}

Equivalently, using
\[
    I(W;X)
    =
    \frac1M
    \inf_{Q_X}
    \sum_{j=1}^M
    D_{\mathrm{KL}}(P_j\|Q_X),
\]
if, for some $\varepsilon\in(0,1)$,
\begin{equation}
    \frac{
    \frac1M
    \inf_{Q_X}
    \sum_{j=1}^M
    D_{\mathrm{KL}}(P_j\|Q_X)
    +
    \log\left(2-\frac1M\right)
    }{
    \log M
    }
    \leq
    1-\varepsilon,
\end{equation}
then, for every $\delta\in(0,\varepsilon)$,
\begin{equation}
    \frac{
    I(W;X)+\log\left(2-\frac1M\right)
    }{
    \log M
    }
    <
    1-\delta,
\end{equation}
and consequently
\begin{equation}
    \mathcal M_{-}(\delta)\geq\rho
    \qquad\text{and}\qquad
    \mathcal M(\delta)\geq\rho.
\end{equation}
This recovers the Fano-type minimax-quantile lower bound~\cite[Lemma 7]{MaVerchandSamworth2024}, with $\rho$ corresponding to the loss threshold
induced by their packing construction.

The preceding exact condition
\[
    I(W;X)
    <
    d_{\mathrm{KL}}
    \left(
    1-\delta
    \middle\|
    \frac1M
    \right)
\]
is slightly sharper, since~\cite[Lemma 7]{MaVerchandSamworth2024} follows after
the relaxation
\[
    d_{\mathrm{KL}}
    \left(
    1-\delta
    \middle\|
    \frac1M
    \right)
    \geq
    (1-\delta)\log M
    -
    \log\left(2-\frac1M\right).
\]
\end{remark}
The approach highlighted above is extremely general. Different choices of \(f\) yield different types of well-known bounds.~\Cref{rmk:classicalFano} recovers the famous Fano's inequality with $f(t)=t\log t$, while the following Corollary shows how to recover Le Cam's two-point method~\cite{Fano1961,LeCam1973,Yu1997}.
The choice $f(t)=\frac12|t-1|$ yields the total variation distance
\[
\operatorname{TV}(P,Q)
:=\frac12\int\left|\frac{\mathrm dP}{\mathrm d\mu}-\frac{\mathrm dQ}{\mathrm d\mu}\right|\mathrm d\mu,
\]
where $\mu$ is any common dominating measure~\cite{AliSilvey1966,Csiszar1967}.
\begin{corollary}
\label{cor:leCam}
Under the same assumptions as~\Cref{thm:fDivLowerBound}, let
$\theta_1,\theta_2\in\Theta$ and let
$P_i\in\mathscr P_{X\mid\theta_i}$, $i\in\{1,2\}$.
Suppose that, for every $\vartheta\in\Theta$,
\begin{equation}
    \left|
    \left\{
    i\in\{1,2\}:
    \ell(\vartheta,\theta_i)<\rho
    \right\}
    \right|
    \leq 1.
\end{equation}
Then
\begin{equation}
    \inf_{\widehat\theta}
    \sup_{\theta\in\Theta}
    \sup_{P_{X\mid\theta}\in\mathscr P_{X\mid\theta}}
    P_{X\mid\theta}
    \left\{
    \ell(\widehat\theta(X),\theta)\geq\rho
    \right\}
    \geq
    \frac{1-\operatorname{TV}(P_1,P_2)}{2}.
\end{equation}
Consequently, if $\delta\in(0,1/2)$ and
\begin{equation}
    \operatorname{TV}(P_1,P_2)<1-2\delta,
\end{equation}
then
\begin{equation}
    \mathcal M_{-}(\delta)\geq\rho
    \qquad\text{and}\qquad
    \mathcal M(\delta)\geq\rho.
\end{equation}
\end{corollary}

\begin{remark}
When the loss is of the form
\[
    \ell(\vartheta,\theta)=g\bigl(d(\vartheta,\theta)\bigr),
\]
with $g$ non-decreasing, let
\[
    \eta:=\frac12d(\theta_1,\theta_2).
\]
The triangle inequality implies that the two success sets at threshold
$g(\eta)$ are disjoint. Thus, taking
\(
    \rho=g(\eta)
\)
in~\Cref{cor:leCam}, the condition
\[
    \operatorname{TV}(P_1,P_2)<1-2\delta
\]
yields
\begin{equation}
    \mathcal M_{-}(\delta)\geq g(\eta).
\end{equation}
This recovers~\cite[Lemma 5]{MaVerchandSamworth2024}.
\end{remark}
Stepping away from $f$-divergences, we can recover a lower-bound on the risk involving R\'enyi divergences and Sibson mutual information~\cite{EspositoVandenbroucqueGastpar2024,EspositoGastparIssa2025}.
Said result also follows from the Neyman--Pearson converse in~\Cref{thm:npLowerBound} and the data-processing
inequality for the binary R\'enyi divergence.

\begin{corollary}
\label{cor:renyiSibson}
Under the same assumptions as~\Cref{thm:npLowerBound}, let $\alpha>1$. Then
\begin{equation}
P_{WX}\bigl(\ell(\widehat\theta(X),W)\geq\rho\bigr)
\geq
1-
\inf_{Q_X}
d_\alpha^{-1}
\left(
D_\alpha(P_{WX}\|P_WQ_X)
\middle\|
L_W(\rho)
\right),
\end{equation}
where $d_\alpha^{-1}(\cdot\|u)$ denotes the generalised inverse of
$p\mapsto d_\alpha(p\|u)$ over $p\in[u,1]$.

Consequently, if $\delta\in(0,1)$, $1-\delta\geq L_W(\rho)$, and
\begin{equation}
I_\alpha(W,X)
<
d_\alpha
\left(
1-\delta
\middle\|
L_W(\rho)
\right),
\end{equation}
then
\begin{equation}
\mathcal M_{-}(\delta)\geq\rho
\qquad\text{and}\qquad
\mathcal M(\delta)\geq\rho.
\end{equation}
\end{corollary}

\begin{remark}
\label{rmk:holderSibsonBound}
A simpler, but weaker, form follows from
\begin{align}
d_\alpha(p|u)&=\frac{1}{\alpha-1}\log\left[p^\alpha u^{1-\alpha}+(1-p)^\alpha(1-u)^{1-\alpha}\right]\geq\frac{1}{\alpha-1}\log p^\alpha u^{1-\alpha}.
\end{align}

Consequently,
\begin{equation}
P_{WX}\bigl(\ell(\widehat\theta(X),W)<\rho\bigr)\leq L_W(\rho)^{\frac{\alpha-1}{\alpha}}\exp\left(\frac{\alpha-1}{\alpha}I_\alpha(W,X)\right).
\end{equation}
This displayed inequality is~\cite[Theorem~8]{EspositoVandenbroucqueGastpar2024}. The preceding bound instead uses the full binary R\'enyi divergence $d_\alpha(p\|q)$.

Thus, if
\begin{equation}
L_W(\rho)^{\frac{\alpha-1}{\alpha}}\exp\left(\frac{\alpha-1}{\alpha}I_\alpha(W,X)\right)<1-\delta,
\end{equation}
then
\[\mathcal M_{-}(\delta)\geq\rho\qquad\text{and}\qquad\mathcal M(\delta)\geq\rho.\]

The binary-R\'enyi formulation above is sharper, since it retains the full binary divergence rather than discarding the complementary term.
\end{remark}

The following Amemiya bound also follows directly from the Neyman--Pearson
converse in~\Cref{thm:npLowerBound}, using the generalised H\"older
inequality to control the estimator-induced success event.

\begin{corollary}
\label{cor:amemiyaHolder}
Under the same assumptions as~\Cref{thm:npLowerBound}, let
$\Phi$ and $\Psi$ be complementary Young functions. Then
\begin{equation}
P_{WX}\bigl(\ell(\widehat\theta(X),W)\geq\rho\bigr)
\geq
1-
\inf_{Q_X}
\frac{
\left\|
\frac{\mathrm dP_{WX}}{\mathrm d(P_WQ_X)}
\right\|_{L_\Phi^{\mathrm A}(P_WQ_X)}
}{
\Psi^{-1}\left(1/L_W(\rho)\right)
}.
\end{equation}
Consequently, if $\delta\in(0,1)$ and
\begin{equation}
\inf_{Q_X}
\frac{
\left\|
\frac{\mathrm dP_{WX}}{\mathrm d(P_WQ_X)}
\right\|_{L_\Phi^{\mathrm A}(P_WQ_X)}
}{
\Psi^{-1}\left(1/L_W(\rho)\right)
}
<
1-\delta,
\end{equation}
then
\begin{equation}
\mathcal M_{-}(\delta)\geq\rho
\qquad\text{and}\qquad
\mathcal M(\delta)\geq\rho.
\end{equation}
\end{corollary}

\begin{remark}
\label{rmk:holderLowerBound}
Under the same assumptions as~\Cref{thm:npLowerBound}, let $\alpha>1$ and let $\alpha'=\alpha/(\alpha-1)$ denote its H\"older conjugate. Then
\begin{equation}
P_{WX}\bigl(\ell(\widehat\theta(X),W)\geq\rho\bigr)
\geq
1-
L_W(\rho)^{\frac{\alpha-1}{\alpha}}
\inf_{Q_X}
\left\|
\frac{\mathrm dP_{WX}}{\mathrm d(P_WQ_X)}
\right\|_{L^\alpha(P_WQ_X)}.
\end{equation}
Equivalently,
\begin{equation}
P_{WX}\bigl(\ell(\widehat\theta(X),W)\geq\rho\bigr)
\geq
1-
L_W(\rho)^{\frac{\alpha-1}{\alpha}}
\exp\left\{
\frac{\alpha-1}{\alpha}
I_\alpha(W,X)
\right\}.
\end{equation}
Consequently, if $\delta\in(0,1)$ and
\begin{equation}
L_W(\rho)^{\frac{\alpha-1}{\alpha}}
\exp\left\{
\frac{\alpha-1}{\alpha}
I_\alpha(W,X)
\right\}
<
1-\delta,
\end{equation}
then
\begin{equation}
\mathcal M_{-}(\delta)\geq\rho
\qquad\text{and}\qquad
\mathcal M(\delta)\geq\rho.
\end{equation}
The second bound above recovers the H\"older relaxation in
\Cref{rmk:holderSibsonBound}. Thus, the same relaxation can be obtained in
two equivalent ways: either by lower-bounding
the binary R\'enyi divergence through
\[
d_\alpha(p\|u)
\geq
\frac{1}{\alpha-1}\log(p^\alpha u^{1-\alpha}),
\]
or directly from H\"older's inequality applied to the Radon--Nikodym
derivative. The binary-R\'enyi bound in~\Cref{cor:renyiSibson} is sharper
because it retains both terms of the binary divergence, although it is less
amenable to explicit computation.
\end{remark}
Letting $\alpha\to\infty$ in~\Cref{cor:renyiSibson} gives the following
result, which may be of independent interest.
\begin{corollary}
    \label{cor:maximalLeakageQuantile}
    Under the same assumptions as~\Cref{thm:npLowerBound},
    \begin{equation}
        P_{WX}(\ell(\hat \theta(X),W)<\rho)\leq L_W(\rho) \exp\{\ml{W}{X}\}. \label{eq:maximalLeakageSuccessBound}
    \end{equation}
    Consequently, if $\delta \in (0,1)$ and 
    \begin{equation}
        L_W(\rho)\exp\{\ml{W}{X}\} < 1-\delta
    \end{equation}
    then
    \begin{equation}
\mathcal M_{-}(\delta)\geq\rho
\qquad\text{and}\qquad
\mathcal M(\delta)\geq\rho.
\end{equation}
\end{corollary}

\begin{figure}[t]
\centering
\resizebox{0.98\linewidth}{!}{%
\begin{tikzpicture}[
    node distance=7mm and 11mm,
    box/.style={
        draw,
        rounded corners,
        align=center,
        inner sep=4pt,
        font=\scriptsize
    },
    equiv/.style={{Latex[length=1.7mm]}-{Latex[length=1.7mm]}, thick},
    arr/.style={-{Latex[length=1.7mm]}, thick},
    relax/.style={-{Latex[length=1.7mm]}, thick, dashed},
    lab/.style={font=\tiny, align=center}
]

\node[box] (NP)
{Neyman--Pearson converse\\\Cref{thm:npLowerBound}};

\node[box, below left=of NP] (fdiv)
{$f$-divergence binary inverse\\\Cref{thm:fDivLowerBound}};

\node[box, below right=of NP] (amemiya)
{Amemiya--H\"older\\\Cref{cor:amemiyaHolder}};

\node[box, below=of fdiv] (finf)
{$f$-informativity\\\Cref{thm:fDivLowerBound}};

\node[box, below left=of finf] (ffano)
{$f$-Fano, $M$-point packing\\\Cref{cor:fFano}};

\node[box, below right=of finf] (renyi)
{R\'enyi binary inverse\\\Cref{cor:renyiSibson}};

\node[box, below=of renyi] (sibson)
{Sibson \(I_\alpha\) converse\\\Cref{cor:renyiSibson}};

\node[box, below=of sibson] (explicit)
{H\"older relaxation\\in terms of \(I_\alpha\)\\\Cref{rmk:holderLowerBound}};

\node[box, right=of explicit] (leakage)
{Maximal Leakage\\\Cref{cor:maximalLeakageQuantile}};

\node[box, below=of amemiya] (holder)
{H\"older $L_\alpha$ bound\\\Cref{rmk:holderLowerBound}};

\node[box, below=of ffano] (fano)
{Classical Fano inequality\\\Cref{rmk:classicalFano}};

\node[box, left=of fano] (lecam)
{Le Cam's two-point method\\\Cref{cor:leCam}};

\draw[arr] (NP) -- (fdiv);
\draw[arr] (NP) -- (amemiya);

\draw[arr] (fdiv) -- (finf);
\draw[arr] (finf) -- (ffano);
\draw[arr] (finf) -- (renyi);

\draw[arr]
(ffano) --
node[left, lab] {KL}
(fano);

\draw[arr]
(ffano.west) -|
node[pos=0.25, above, lab] {$M=2$, TV}
(lecam.north);

\draw[arr] (renyi) -- (sibson);

\draw[relax]
(sibson) --
(explicit);

\draw[arr]
(explicit) --
node[above, lab] {$\alpha\to\infty$}
(leakage);

\draw[arr] (amemiya) -- (holder);
\draw[equiv] (holder.south west) -- (explicit.north east);

\node[
    below=11mm of fano,
    anchor=west,
    font=\tiny
] (legend1)
{\tikz\draw[arr] (0,0)--(0.7,0);
 \quad specialisation / exact implication};

\node[
    below=2.5mm of legend1.west,
    anchor=west,
    font=\tiny
] (legend2)
{\tikz\draw[relax] (0,0)--(0.7,0);
 \quad relaxation};

\node[
    below=3mm of legend2.west,
    anchor=west,
    font=\tiny
] (legend3)
{\tikz\draw[equiv] (0,0)--(0.7,0);
 \quad equivalent formulations};

\end{tikzpicture}%
}

\caption{
Hierarchy of the minimax-quantile lower bounds. Solid arrows denote
specialisations or exact implications; dashed arrows denote
relaxations. The H\"older relaxation in terms of \(I_\alpha\) follows either by
relaxing the binary R\'enyi bound or directly from H\"older's
inequality. Letting $\alpha\to\infty$ gives the Maximal Leakage bound.
}
\label{fig:lower-bound-hierarchy}
\end{figure}
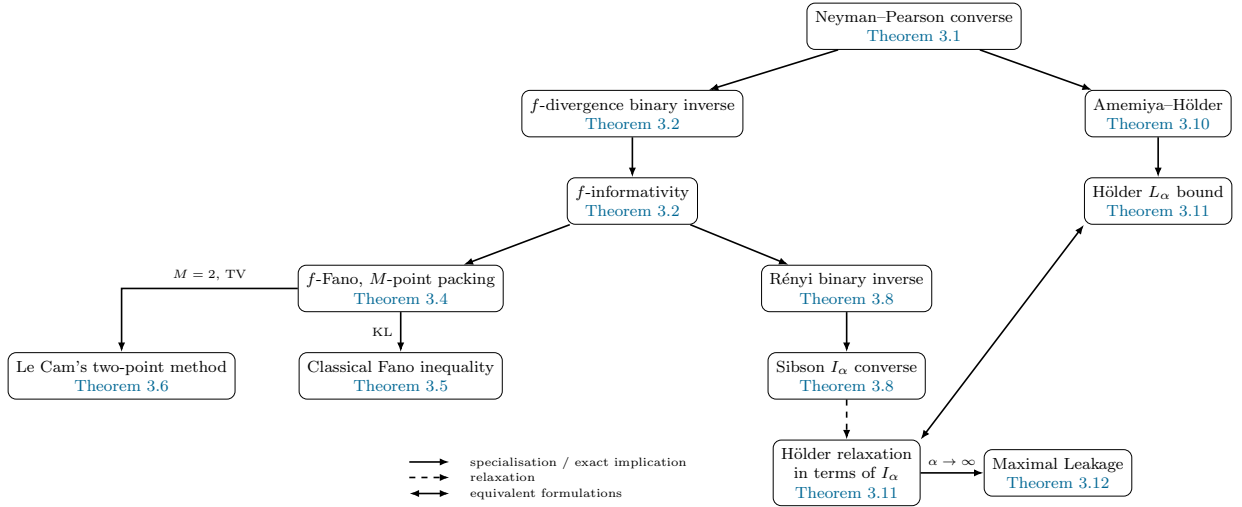
\section{Applications}
\subsection{Exact Recovery and Maximal Leakage}
\label{sec:exact-recovery-maximal-leakage}
The following result identifies a class of symmetric problems in which the
Maximal Leakage converse coincides with both the Neyman--Pearson converse
and the exact minimax risk.
\begin{theorem}
\label{thm:maximalLeakageExact}
Let
\(\mathcal V=\{\theta_1,\ldots,\theta_M\}\subseteq\Theta\)
be a finite loss-ball packing at threshold $\rho$, meaning that
\begin{equation}
\sup_{\vartheta\in\Theta}
\left|
\left\{j\in[M]:\ell(\vartheta,\theta_j)<\rho\right\}
\right|
\leq1.
\label{eq:lossBallPackingDefinition}
\end{equation}
Let $W$ be uniformly distributed over $\mathcal V$, with
\(P_{X\mid W=\theta_j}=P_{X\mid\theta_j}\) for \(j=1,\ldots,M\).
In particular, \(L_W(\rho)=1/M\).

Suppose that the uniform-prior MAP estimator
$\widehat\theta_{\mathrm{MAP}}$ is an equaliser on $\mathcal V$, namely,
for some $s_{\mathrm{MAP}}\in[0,1]$,
\begin{equation}
P_{X\mid\theta_j}
\left(
\widehat\theta_{\mathrm{MAP}}(X)=\theta_j
\right)
=
s_{\mathrm{MAP}}
\qquad
\text{for every }j=1,\ldots,M.
\label{eq:mapEqualizer}
\end{equation}
Then
\begin{align}
&\inf_{\widehat\theta\in\mathcal E}
\max_{j=1,\ldots,M}
P_{X\mid\theta_j}
\left(
\ell\bigl(\widehat\theta(X),\theta_j\bigr)\geq\rho
\right)
\notag\\
&\qquad=
1-
P_{WX}
\left(
\widehat\theta_{\mathrm{MAP}}(X)=W
\right)
\label{eq:minimaxMapEquality}\\
&\qquad=
1-
\inf_{Q_X}
\beta_{\circ}(P_{WX},P_WQ_X)^{-1}
\left(\frac1M\right)
\label{eq:minimaxNPEquality}\\
&\qquad=
1-
\frac{\exp\{\ml{W}{X}\}}{M}.
\label{eq:minimaxLeakageEquality}
\end{align}
In particular,
\begin{equation}
P_{WX}
\left(
\widehat\theta_{\mathrm{MAP}}(X)=W
\right)
=
\inf_{Q_X}
\beta_{\circ}(P_{WX},P_WQ_X)^{-1}
\left(\frac1M\right)
=
\frac{\exp\{\ml{W}{X}\}}{M}.
\label{eq:mapNPLeakageEquality}
\end{equation}
\end{theorem}

For a uniform finite prior, the MAP-success identity in~\eqref{eq:mapNPLeakageEquality} follows from the closed form for Maximal Leakage~\cite{IssaWagnerKamath2020}. Under the equaliser assumption, the average MAP success probability equals its success probability under every parameter in the packing, which gives the first equality in~\eqref{eq:minimaxMapEquality}.
The next example showcases an application of this principle to provide finite sample rates for community detection in Stochastic Block Models.
\subsubsection{Exact Recovery: Gaussian weighted stochastic block model}
\label{sec:gaussian-weighted-sbm}
Stochastic block models form a standard class of latent-variable models for
community detection in networks~\cite{HollandLaskeyLeinhardt1983}. Each vertex is assigned an unobserved
community label, and, conditionally on these labels, the edge observations
are independent, with distributions depending only on whether their
two vertices belong to the same community. The resulting inference problem is
to recover the latent partition from the observed network.

While the classical stochastic block model records only the presence or
absence of an edge, many applications provide real-valued similarities,
correlations, or interaction strengths. Weighted stochastic block models
accommodate such observations by assigning different edge-weight
distributions to within-community and between-community pairs~\cite{AicherJacobsClauset2015}. The Gaussian
model considered below is a particularly tractable representative of this
class: the difference between the two Gaussian means determines the signal
strength, while the common variance controls the noise level. The Gaussian specification provides an analytically tractable model without making the recovery problem trivial. Although the likelihood difference between the true labelling and any fixed alternative labelling is Gaussian, exact recovery requires controlling these differences jointly over an exponentially large and strongly dependent family of alternative labellings. The Gaussian structure nevertheless makes these comparisons sufficiently explicit to derive finite-sample bounds and to study the critical recovery window beyond the known first-order threshold.
We focus on a balanced partition into two communities of equal size. Since the symbols \(+1\) and \(-1\) merely name the two communities, the labellings \(\sigma\) and \(-\sigma\) represent the same partition. We therefore identify labellings that differ only by a global sign
reversal. More formally, let $n=2m$ be even and define
\(\widetilde\Theta_n:=\{\sigma\in\{-1,+1\}^n:
\langle\sigma,\mathbf 1_n\rangle=0\}\).
Since $\sigma$ and $-\sigma$ represent the same partition, the
identifiable parameter space is
\(\Theta_n:=\widetilde\Theta_n/\!\sim\), where \(\sigma\sim-\sigma\),
and we denote the equivalence class of $\sigma$ by $[\sigma]$. In
particular, \(|\Theta_n|=\frac12\binom{n}{n/2}\).

For $[\sigma]\in\Theta_n$, let $P_\sigma$ denote the distribution of
the symmetric weighted adjacency matrix
\(A=(A_{ij})_{i,j\in[n]}\), with \(A_{ii}=0\),
whose upper-triangular entries are conditionally independent and
satisfy
\begin{equation}
A_{ij}\mid[\sigma]
\sim
\begin{cases}
\mathcal N(\mu_{1,n},\tau^2),
&
\sigma_i=\sigma_j,
\\[1mm]
\mathcal N(\mu_{2,n},\tau^2),
&
\sigma_i\neq\sigma_j,
\end{cases}
\qquad
1\leq i<j\leq n,
\label{eq:gwsbmModel}
\end{equation}
where $\mu_{1,n}>\mu_{2,n}$ and $\tau>0$. For notational convenience, let
\(\overline\mu_n:=(\mu_{1,n}+\mu_{2,n})/2\) and
\(\Delta_n:=\mu_{1,n}-\mu_{2,n}\).
Equivalently,
\begin{equation}
A_{ij}
=
\overline\mu_n
+
\frac{\Delta_n}{2}\sigma_i\sigma_j
+
\tau Z_{ij},
\qquad
Z_{ij}\overset{\mathrm{iid}}{\sim}\mathcal N(0,1).
\label{eq:gwsbmSignalNoiseForm}
\end{equation} 
The distribution in~\eqref{eq:gwsbmModel} is independent of the
choice of representative of $[\sigma]$.

An estimator is a measurable map
\(\widehat\sigma:\mathbb R^{n\times n}\to\Theta_n\).
Under exact-recovery loss, its minimax risk is
\begin{equation}
\mathcal R_n^\star
:=
\inf_{\widehat\sigma}
\sup_{[\sigma]\in\Theta_n}
P_\sigma
\left(
\widehat\sigma(A)\neq[\sigma]
\right).
\label{eq:gwsbmMinimaxRisk}
\end{equation}

Under the uniform distribution on $\Theta_n$, the MAP estimator
coincides with the maximum-likelihood estimator and can be written as
\begin{equation}
\widehat\sigma_{\mathrm{ML}}(A)
\in
\operatorname*{arg\,max}_{[\sigma]\in\Theta_n}
\sum_{1\leq i<j\leq n}
A_{ij}\sigma_i\sigma_j.
\label{eq:gwsbmMLE}
\end{equation}
Indeed, up to additive constants and multiplication by the positive
factor $\Delta_n/(2\tau^2)$, the log-likelihood is
\[
\sum_{1\leq i<j\leq n}
\left(A_{ij}-\overline\mu_n\right)\sigma_i\sigma_j.
\]
Since every balanced labelling satisfies
\[
\sum_{1\leq i<j\leq n}\sigma_i\sigma_j
=
\frac12
\left[
\left(\sum_{i=1}^n\sigma_i\right)^2-n
\right]
=
-\frac n2,
\]
the term involving $\overline\mu_n$ is constant over $\Theta_n$,
which yields~\eqref{eq:gwsbmMLE}.
The objective in~\eqref{eq:gwsbmMLE} is invariant under
$\sigma\mapsto-\sigma$. Moreover, the action of the permutation group
on $\Theta_n$ is transitive, and the family $\{P_\sigma:\sigma\in\Theta_n\}$
is equivariant under the corresponding permutations of the rows and columns
of $A$. Hence the
uniform-prior MAP rule is an equaliser and is minimax optimal:
\begin{equation}
\mathcal R_n^\star
=
P_\sigma
\left(
\widehat\sigma_{\mathrm{ML}}(A)\neq[\sigma]
\right),
\qquad
[\sigma]\in\Theta_n.
\label{eq:gwsbmMLEMinimax}
\end{equation}

Following~\cite{PandeyKulkarni2024}, we define the signal-to-noise
ratio as follows
\begin{equation}
\operatorname{SNR}_n
:=
\frac{n(\mu_{1,n}-\mu_{2,n})^2}{8\tau^2\log n}.
\label{eq:gwsbmSNR}
\end{equation}
Under the critical scaling
\begin{equation}
\mu_{1,n}
=
\alpha\sqrt{\frac{\log n}{n}},
\qquad
\mu_{2,n}
=
\beta\sqrt{\frac{\log n}{n}},
\qquad
\alpha>\beta,
\label{eq:gwsbmCriticalScaling}
\end{equation}
this quantity is constant:
\(\operatorname{SNR}_n=\operatorname{SNR}
:=(\alpha-\beta)^2/(8\tau^2)\).

The sharp first-order exact-recovery threshold is known to occur at
$\operatorname{SNR}=1$~\cite{PandeyKulkarni2024}. More precisely,
under~\eqref{eq:gwsbmCriticalScaling},
\begin{equation}
\operatorname{SNR}<1
\quad\Longrightarrow\quad
\mathcal R_n^\star\longrightarrow1,
\label{eq:gwsbmKnownImpossibility}
\end{equation}
whereas
\begin{equation}
\operatorname{SNR}>1
\quad\Longrightarrow\quad
\mathcal R_n^\star\longrightarrow0,
\label{eq:gwsbmKnownPossibility}
\end{equation}
and the latter limit is attained by the maximum-likelihood estimator.
Suitable semidefinite and spectral relaxations also attain exact recovery
above the same threshold.  Thus, at first order, the statistically
achievable region is also achieved by known polynomial-time methods.

The appearance of~\eqref{eq:gwsbmSNR} is consistent with the
order-$1/2$ R\'enyi-divergence characterisation of exact recovery in
weighted stochastic block models~\cite{JogLoh2015}. Indeed,
\(D_{1/2}(\mathcal N(\mu_{1,n},\tau^2)\|\mathcal
N(\mu_{2,n},\tau^2))=\Delta_n^2/(4\tau^2)\), and therefore
\(\operatorname{SNR}_n=\frac{n}{2\log n}
D_{1/2}(\mathcal N(\mu_{1,n},\tau^2)\|\mathcal
N(\mu_{2,n},\tau^2))\).

The established results
\eqref{eq:gwsbmKnownImpossibility}--\eqref{eq:gwsbmKnownPossibility}
determine the first-order all-or-nothing transition, but do not
characterise the minimax risk when
\(\operatorname{SNR}_n=1+o(1)\). We call this regime the \emph{critical
window}. The purpose of the following analysis is to obtain explicit
finite-sample bounds and to characterise the minimax exact-recovery risk
throughout that window.
To connect this problem with the preceding information-theoretic
framework, let $W$ be uniformly distributed over $\Theta_n$, and,
conditionally on $W=[\sigma]$, let $A$ have distribution $P_\sigma$.
By permutation invariance, the uniform-prior maximum-likelihood rule is
an equaliser. Consequently, the exact-recovery identity in
\Cref{thm:maximalLeakageExact} gives
\begin{equation}
1-\mathcal R_n^\star
=
P_{WA}
\left(
\widehat\sigma_{\mathrm{ML}}(A)=W
\right)
=
\frac{\exp\{\ml{W}{A}\}}{|\Theta_n|}.
\label{eq:gwsbmLeakageRiskIdentity}
\end{equation}
Thus, obtaining finite-sample guarantees for exact recovery is equivalent
to estimating the normalised exponential
\(\exp\{\ml{W}{A}\}/|\Theta_n|\).
The pointwise supremum in the formula for Maximal Leakage does not factorise
over the edges. Indeed, if $s_{ij}:=\sigma_i\sigma_j$, then
$s_{ij}s_{jk}s_{ik}=1$ for every distinct $i,j,k$, while balance requires
$\sum_{i=1}^n\sigma_i=0$. Thus, the edge signs cannot be optimised
independently. We therefore estimate the leakage by
analysing the likelihood comparisons between the true partition and
balanced exchanges of vertices between the two communities.
Fix $[\sigma^\star]\in\Theta_n$ and choose a representative
$\sigma^\star\in\widetilde\Theta_n$. Define
\(C_+:=\{i\in[n]:\sigma_i^\star=+1\}\) and
\(C_-:=\{i\in[n]:\sigma_i^\star=-1\}\).
Since $\sigma^\star$ is balanced, \(|C_+|=|C_-|=m\).

Define \(\lambda_n:=(\mu_{1,n}-\mu_{2,n})/(2\tau)\).
For every $1\leq i<j\leq n$, let
\(Y_{ij}:=\frac{\sigma_i^\star\sigma_j^\star}{\tau}
(A_{ij}-(\mu_{1,n}+\mu_{2,n})/2)\).
Set \(Y_{ji}:=Y_{ij}\) and \(Y_{ii}:=0\). Under
$P_{\sigma^\star}$, \(Y_{ij}=\lambda_n+Z_{ij}\) for
\(1\leq i<j\leq n\), where
\((Z_{ij})_{1\leq i<j\leq n}\) are independent standard Gaussian
variables.

For every $i\in[n]$, define \(D_i:=\sum_{j\neq i}Y_{ij}\) and
\(R_i:=\{D_i-(n-1)\lambda_n\}/\sqrt{n-2}\). Also define
\(M_{+,n}:=\max_{i\in C_+}(-R_i)\) and
\(M_{-,n}:=\max_{j\in C_-}(-R_j)\).

The centred Gaussian vector $(R_i)_{i\in[n]}$ satisfies
\[
\operatorname{Var}(R_i)
=
\frac{n-1}{n-2},
\qquad
\operatorname{Cov}(R_i,R_j)
=
\frac{1}{n-2},
\qquad
i\neq j.
\]
Let \(G,\xi_1,\ldots,\xi_n\) be mutually independent random variables with distribution
$\mathcal N(0,1)$. The Gaussian vector
\[
\left(
-\xi_i+\frac{G}{\sqrt{n-2}}
\right)_{i\in[n]}
\]
has the same mean vector and covariance matrix as
$(R_i)_{i\in[n]}$. Consequently,
\[
(R_i)_{i\in[n]}
\overset{\mathrm d}{=}
\left(
-\xi_i+\frac{G}{\sqrt{n-2}}
\right)_{i\in[n]}.
\]

Define \(U_m:=\max_{i\in C_+}\xi_i\) and
\(V_m:=\max_{j\in C_-}\xi_j\).
Since $C_+$ and $C_-$ are disjoint, $U_m$ and $V_m$ are independent
copies of the maximum of $m$ independent standard Gaussian random
variables, and both are independent of $G$. It follows jointly that
\[
M_{+,n}
\overset{\mathrm d}{=}
U_m-\frac{G}{\sqrt{n-2}},
\qquad
M_{-,n}
\overset{\mathrm d}{=}
V_m-\frac{G}{\sqrt{n-2}},
\]
and therefore
\[
M_{+,n}+M_{-,n}
\overset{\mathrm d}{=}
U_m+V_m-\frac{2G}{\sqrt{n-2}}.
\]

We are now ready to state the main finite-sample result.

\begin{theorem}
\label{thm:gwsbmFiniteLeakage}
For every integer \(1\leq L\leq\lfloor m/2\rfloor\) and every \(t>0\),
one has
\begin{align}
&
1-
P\left(
U_m+V_m-\frac{2G}{\sqrt{n-2}}
\geq
2\vartheta_n-\varepsilon_{n,L}(t)
\right)
-r_n(t)-B_{n,L}
\notag\\
&\leq
\frac{\exp\{\ml{W}{A}\}}{|\Theta_n|}
\notag\\
&\leq
1-
P\left(
U_m+V_m-\frac{2G}{\sqrt{n-2}}
\geq
2\vartheta_n+\varepsilon_{n,1}(t)
\right)
+r_n(t),
\label{eq:gwsbmFiniteLeakageSandwich}
\end{align}
Here
\[
\begin{aligned}
r_n(t)&:=n(n-1)\overline\Phi(t),\\
B_{n,L}&:=
\sum_{k=L+1}^{\lfloor m/2\rfloor}
\binom{m}{k}^{2}
\overline\Phi\!\left(2\lambda_n\sqrt{k(m-k)}\right),\\
\vartheta_n&:=\frac{(n-1)\lambda_n}{\sqrt{n-2}},
&
\varepsilon_{n,L}(t)&:=
\frac{2(2L-1)(\lambda_n+t)}{\sqrt{n-2}}.
\end{aligned}
\]

Consequently,
\begin{align}
&
P\left(
U_m+V_m-\frac{2G}{\sqrt{n-2}}
\geq
2\vartheta_n+\varepsilon_{n,1}(t)
\right)
-r_n(t)
\notag\\
&\leq
\mathcal R_n^\star
\notag\\
&\leq
P\left(
U_m+V_m-\frac{2G}{\sqrt{n-2}}
\geq
2\vartheta_n-\varepsilon_{n,L}(t)
\right)
+r_n(t)+B_{n,L}.
\label{eq:gwsbmFiniteRiskSandwich}
\end{align}
\end{theorem}

\begin{proof}
For every $[\sigma]\in\Theta_n$, let $p_\sigma$ denote the density of
$P_\sigma$ with respect to Lebesgue measure on
$\mathbb R^{\binom{n}{2}}$. Since $W$ is uniformly distributed over
$\Theta_n$, the exponential of Maximal Leakage is
\begin{equation}
\exp\{\ml{W}{A}\}
=
\int_{\mathbb R^{\binom{n}{2}}}
\max_{[\sigma]\in\Theta_n}
p_\sigma(a)
\,\mathrm da.
\label{eq:gwsbmLeakageIntegral}
\end{equation}

For every $[\sigma]\in\Theta_n$, define its maximum-likelihood region by
\begin{equation}
\mathcal D_\sigma
:=
\left\{
a:
p_\sigma(a)>
p_{\sigma'}(a)
\text{ for every }
[\sigma']\neq[\sigma]
\right\}.
\label{eq:gwsbmMLRegion}
\end{equation}
Likelihood ties have Lebesgue measure zero, and therefore
\begin{align}
\exp\{\ml{W}{A}\}
&=
\sum_{[\sigma]\in\Theta_n}
\int_{\mathcal D_\sigma}
p_\sigma(a)
\,\mathrm da.
\label{eq:gwsbmLeakageDecisionRegions}
\end{align}
The action of the permutation group on $\Theta_n$ is transitive, and
simultaneously permuting the rows and columns of $A$ maps
$\mathcal D_\sigma$ onto the corresponding maximum-likelihood region.
Consequently, all the integrals in
\eqref{eq:gwsbmLeakageDecisionRegions} are equal. For the fixed
$[\sigma^\star]\in\Theta_n$,
\begin{equation}
\frac{\exp\{\ml{W}{A}\}}{|\Theta_n|}
=
\int_{\mathcal D_{\sigma^\star}}
p_{\sigma^\star}(a)
\,\mathrm da
=
P_{\sigma^\star}
\left(
A\in\mathcal D_{\sigma^\star}
\right).
\label{eq:gwsbmNormalisedLeakageRegion}
\end{equation}
Thus, the normalised exponential of Maximal Leakage is the
$P_{\sigma^\star}$-measure of the region on which the true labelling
maximises the likelihood.

We now characterise this region through balanced exchanges. Let
\(U\subseteq C_+\) and \(V\subseteq C_-\) satisfy \(|U|=|V|=k\),
and let $\sigma^{U,V}$ be the labelling obtained from
$\sigma^\star$ by reversing the signs of the vertices in $U\cup V$.
Define
\begin{equation}
T_{U,V}
:=
\sum_{\substack{1\leq i<j\leq n\\
|\{i,j\}\cap(U\cup V)|=1}}
Y_{ij}.
\label{eq:gwsbmExchangeStatistic}
\end{equation}
A direct calculation of the Gaussian likelihood ratio gives
\begin{equation}
\log
\frac{
p_{\sigma^{U,V}}(A)
}{
p_{\sigma^\star}(A)
}
=
-2\lambda_n T_{U,V}.
\label{eq:gwsbmExchangeLikelihoodRatio}
\end{equation}
Since $\lambda_n>0$, the competing labelling has likelihood at least
as large as the true labelling if and only if
\begin{equation}
T_{U,V}\leq0.
\label{eq:gwsbmDetrimentalExchange}
\end{equation}
Every element of $\Theta_n$ can be represented in this way, and,
because a labelling and its global sign reversal are identified, it is
enough to consider \(1\leq k\leq\lfloor m/2\rfloor\).
Up to an event of probability zero,
\begin{equation}
\mathcal D_{\sigma^\star}
=
\bigcap_{k=1}^{\lfloor m/2\rfloor}
\bigcap_{\substack{U\subseteq C_+,\,V\subseteq C_-\\
|U|=|V|=k}}
\left\{
T_{U,V}>0
\right\}.
\label{eq:gwsbmMLRegionIntersections}
\end{equation}

Let
\begin{equation}
\mathcal E_t
:=
\left\{
\max_{1\leq i<j\leq n}|Z_{ij}|\leq t
\right\}.
\label{eq:gwsbmBoundedNoiseEvent}
\end{equation}
A union bound gives
\begin{align}
P_{\sigma^\star}
\left(
\mathcal E_t^c
\right)
&\leq
2\binom{n}{2}\overline\Phi(t)\\
&=
n(n-1)\overline\Phi(t)\\
&=
r_n(t).
\label{eq:gwsbmBoundedNoiseProbability}
\end{align}

We first derive an upper bound on
$\exp\{\ml{W}{A}\}/|\Theta_n|$. Let
$i^\star\in C_+$ and $j^\star\in C_-$ satisfy
\(-R_{i^\star}=M_{+,n}\) and \(-R_{j^\star}=M_{-,n}\).
For the exchange of $i^\star$ and $j^\star$,
\begin{align}
T_{\{i^\star\},\{j^\star\}}
&=
D_{i^\star}+D_{j^\star}-2Y_{i^\star j^\star}\\
&=
\sqrt{n-2}
\left(
2\vartheta_n-M_{+,n}-M_{-,n}
\right)
-
2\left(
\lambda_n+Z_{i^\star j^\star}
\right).
\label{eq:gwsbmMostDetrimentalOneExchange}
\end{align}
If \(M_{+,n}+M_{-,n}\geq
2\vartheta_n+\varepsilon_{n,1}(t)\) and $\mathcal E_t$ occurs, then
\begin{align}
T_{\{i^\star\},\{j^\star\}}
&\leq
-2(\lambda_n+t)
-
2(\lambda_n-t)\\
&=
-4\lambda_n
<
0.
\end{align}
It follows from~\eqref{eq:gwsbmMLRegionIntersections} that
\begin{equation}
\left\{
M_{+,n}+M_{-,n}
\geq
2\vartheta_n+\varepsilon_{n,1}(t)
\right\}
\cap
\mathcal E_t
\subseteq
\mathcal D_{\sigma^\star}^c.
\label{eq:gwsbmLeakageUpperInclusion}
\end{equation}
Using~\eqref{eq:gwsbmNormalisedLeakageRegion} and
\eqref{eq:gwsbmBoundedNoiseProbability},
\begin{align}
\frac{\exp\{\ml{W}{A}\}}{|\Theta_n|}
&=
P_{\sigma^\star}
\left(
A\in\mathcal D_{\sigma^\star}
\right)\\
&\leq
1-
P_{\sigma^\star}
\left(
M_{+,n}+M_{-,n}
\geq
2\vartheta_n+\varepsilon_{n,1}(t)
\right)
+r_n(t).
\label{eq:gwsbmLeakageUpperBound}
\end{align}

We next derive the lower bound. Suppose that
$T_{U,V}\leq0$ for some \(U\subseteq C_+\) and \(V\subseteq C_-\)
with \(|U|=|V|=k\leq L\).
Writing $U\cup V$ for the set of $2k$ exchanged vertices, one has
\begin{align}
\sum_{i\in U\cup V}D_i
&=
T_{U,V}
+
2
\sum_{\substack{1\leq i<j\leq n\\
i,j\in U\cup V}}
Y_{ij}.
\label{eq:gwsbmVertexScoreDecomposition}
\end{align}
On $\mathcal E_t$, the second sum contains
\(\binom{2k}{2}=k(2k-1)\) terms, each bounded above by
$\lambda_n+t$. Since $T_{U,V}\leq0$, it follows that
\(\sum_{i\in U\cup V}D_i\leq2k(2k-1)(\lambda_n+t)\).
Using \(D_i=(n-1)\lambda_n+\sqrt{n-2}\,R_i\),
we obtain
\begin{equation}
-\sum_{i\in U\cup V}R_i
\geq
2k\vartheta_n
-
\frac{
2k(2k-1)(\lambda_n+t)
}{
\sqrt{n-2}
}.
\label{eq:gwsbmDetrimentalSmallExchangeLower}
\end{equation}
Moreover,
\(-\sum_{i\in U\cup V}R_i\leq kM_{+,n}+kM_{-,n}\).
Consequently, every detrimental exchange with $k\leq L$, on
$\mathcal E_t$, implies
\begin{align}
M_{+,n}+M_{-,n}
&\geq
2\vartheta_n
-
\frac{
2(2k-1)(\lambda_n+t)
}{
\sqrt{n-2}
}\\
&\geq
2\vartheta_n-\varepsilon_{n,L}(t).
\label{eq:gwsbmSmallExchangeImplication}
\end{align}

It remains to control exchanges with $k>L$. For fixed
$U\subseteq C_+$ and $V\subseteq C_-$ satisfying
$|U|=|V|=k$, the statistic $T_{U,V}$ is the sum of
$4k(m-k)$ independent random variables with distribution
$\mathcal N(\lambda_n,1)$. Hence
\(P_{\sigma^\star}(T_{U,V}\leq0)
=\overline\Phi(2\lambda_n\sqrt{k(m-k)})\).
There are $\binom{m}{k}^2$ possible pairs $(U,V)$. Therefore,
\begin{align}
&P_{\sigma^\star}
\left(
T_{U,V}\leq0
\text{ for some }
U,V
\text{ with }
L<|U|=|V|
\leq
\left\lfloor\frac m2\right\rfloor
\right)\\
&\qquad\leq
\sum_{k=L+1}^{\lfloor m/2\rfloor}
\binom{m}{k}^2
\overline\Phi
\left(
2\lambda_n\sqrt{k(m-k)}
\right)\\
&\qquad=
B_{n,L}.
\label{eq:gwsbmLargeExchangeProbability}
\end{align}

Combining~\eqref{eq:gwsbmMLRegionIntersections},
\eqref{eq:gwsbmSmallExchangeImplication}, and
\eqref{eq:gwsbmLargeExchangeProbability}, one obtains
\begin{align}
\mathcal D_{\sigma^\star}^c
\subseteq{}&
\left\{
M_{+,n}+M_{-,n}
\geq
2\vartheta_n-\varepsilon_{n,L}(t)
\right\}
\cup
\mathcal E_t^c
\notag\\
&\cup
\left\{
T_{U,V}\leq0
\text{ for some }
U,V
\text{ with }
L<|U|=|V|
\leq
\left\lfloor\frac m2\right\rfloor
\right\}.
\end{align}
Using~\eqref{eq:gwsbmNormalisedLeakageRegion},
\eqref{eq:gwsbmBoundedNoiseProbability}, and
\eqref{eq:gwsbmLargeExchangeProbability}, it follows that
\begin{align}
\frac{\exp\{\ml{W}{A}\}}{|\Theta_n|}
&\geq
1-
P_{\sigma^\star}
\left(
M_{+,n}+M_{-,n}
\geq
2\vartheta_n-\varepsilon_{n,L}(t)
\right)
\notag\\
&\qquad
-r_n(t)-B_{n,L}.
\label{eq:gwsbmLeakageLowerBound}
\end{align}

Finally, using
\(M_{+,n}+M_{-,n}\overset{\mathrm d}{=}
U_m+V_m-2G/\sqrt{n-2}\)
in~\eqref{eq:gwsbmLeakageUpperBound} and
\eqref{eq:gwsbmLeakageLowerBound} proves
\eqref{eq:gwsbmFiniteLeakageSandwich}.

Since the maximum-likelihood estimator is an equaliser and minimax,
\(\mathcal R_n^\star
=1-\exp\{\ml{W}{A}\}/|\Theta_n|\).
Taking complements in the Maximal Leakage bounds gives
\eqref{eq:gwsbmFiniteRiskSandwich}.
\end{proof}
To the best of our knowledge,~\Cref{thm:gwsbmFiniteLeakage} provides
the first explicitly stated two-sided finite-sample bounds on the exact-recovery minimax risk
$\mathcal R_n^\star$ in the Gaussian weighted stochastic block model.
Previous results characterise the asymptotic exact-recovery threshold
through the order-$1/2$ R\'enyi divergence or the signal-to-noise ratio,
but do not provide a computable two-sided nonasymptotic estimate  of the minimax risk in the form given
above~\cite{JogLoh2015,PandeyKulkarni2024}. The finite-sample bounds in \Cref{thm:gwsbmFiniteLeakage} immediately
recover the sharp first-order exact-recovery transition.

\begin{corollary}
\label{cor:sbmStrongConverse}
Suppose that \(\operatorname{SNR}_n\to\gamma\neq1\). Then
\[
R_n^\star\longrightarrow
\begin{cases}
1, & \gamma<1,\\
0, & \gamma>1.
\end{cases}
\]
\end{corollary}

The result represents a strong converse: when \(\gamma<1\), the
optimal probability of exact recovery \(1-R_n^\star\) converges to zero.
When \(\gamma>1\), the maximum-likelihood estimator achieves exact
recovery with probability tending to one. The example we just explored represents a detection problem. 
The next example employs instead Maximal Leakage in a continuous estimation setting.
\subsubsection{Continuous Estimation: Low-rank denoising with isotropic bounded-energy noise}
\label{sec:boundedEnergyLowRank}

Low-rank matrix denoising asks for estimation of an unknown matrix from a noisy
version of all its entries, under the structural assumption that its effective
dimension is much smaller than the ambient dimension. In the canonical
Gaussian model, singular-value shrinkage and thresholding are known to attain
sharp or asymptotically sharp Frobenius risks; see, for example,
\cite{DonohoGavish2014}. Uniform laws on matrix $p$-balls arise in asymptotic
convex geometry~\cite{KaufmannThale2022}, and rotationally invariant matrix
noise has also been considered in asymptotic spiked models
\cite{DudejaLiuMa2026}. The model below combines these two themes: the signal
is an arbitrary low-rank matrix, while the noise is isotropic and supported
on a bounded Frobenius ball.

Let \(D:=d_1d_2\), \(1\leq r\leq d_1\wedge d_2\), \(\sigma>0\), and
\(\Theta_r:=\{M\in\mathbb R^{d_1\times d_2}:
\operatorname{rank}(M)\leq r\}\).
We observe
\begin{equation}
    Y=M+\tau_D U,
    \qquad
    \tau_D:=\sigma\sqrt{D+2},
    \qquad
    U\sim\operatorname{Unif}\{A:\|A\|_F\leq1\},
    \qquad
    M\in\Theta_r.
\label{eq:boundedEnergyLowRankModel}
\end{equation}
After vectorisation, $U$ is uniform on the Euclidean unit ball in
$\mathbb R^D$. In particular,
\(\operatorname{Cov}(\operatorname{vec}(\tau_DU))=\sigma^2I_D\),
so $\sigma^2$ has the same per-coordinate variance interpretation as in white
Gaussian matrix denoising. Unlike white Gaussian noise, however, the entries
are dependent, the law is rotationally invariant, and the perturbation
satisfies \(\|Y-M\|_F\leq\sigma\sqrt{D+2}\) almost surely.
Equivalently, conditional on $M$, the observation $Y$ is uniformly distributed
on the Frobenius ball
\(B_F(M,\tau_D):=\{A\in\mathbb R^{d_1\times d_2}:
\|A-M\|_F\leq\tau_D\}\).
We denote this distribution by $P_M$.

Standard information-theoretic lower bounds for high-dimensional matrix
estimation usually separate two difficulties. The first is the size of the
parameter space: one constructs a finite Frobenius packing of low-rank
matrices, reduces estimation to multiple testing, and applies Fano's
inequality after controlling pairwise Kullback--Leibler or R\'enyi
divergences. The second is the confidence level, which is commonly handled
by a two-point Le Cam argument. For additive white Gaussian noise, writing
\(P_M^{\mathrm G}:=
\mathcal N_D(\operatorname{vec}M,\sigma^2I_D)\),
the first step is particularly effective because
\(D_{\mathrm{KL}}(P_M^{\mathrm G}\|P_{M'}^{\mathrm G})
=\|M-M'\|_F^2/(2\sigma^2)\),
so the packing geometry translates directly into a bound depending on the
effective number of free parameters.

The bounded-support channel in
\eqref{eq:boundedEnergyLowRankModel} behaves differently. Distinct signal
matrices generate uniform distributions on translated Frobenius balls with
nonidentical supports. Consequently,
\begin{equation}
D_{\mathrm{KL}}(P_M\|P_{M'})=\infty,
\qquad
D_\alpha(P_M\|P_{M'})=\infty,
\quad \alpha>1,
\label{eq:lowRankPairwiseDivergencesInfinite}
\end{equation}
for every $M\neq M'$. Hence the usual pairwise-KL or pairwise-R\'enyi packing
argument is vacuous. In contrast, the Maximal Leakage associated with a
compact continuous prior is finite and reduces to the volume of a tubular
neighbourhood of the prior support, yielding a nonvacuous minimax-quantile
lower bound.

The goal is to estimate the full signal matrix $M$. We measure performance using the
normalised squared Frobenius loss
\begin{equation}
    \ell_F(\widehat M,M)
    :=
    \frac1D\|\widehat M-M\|_F^2.
\label{eq:boundedEnergyLowRankLoss}
\end{equation}
Thus $\ell_F(\widehat M,M)$ is the average squared estimation error per
matrix entry. We apply the definitions of
\Cref{def:minimax-quantile,def:lower-minimax-quantile} with
\(\Theta=\Theta_r\), \(X=Y\), \(P_{X\mid\theta}=P_M\), and
\(\ell=\ell_F\).
Accordingly, $\mathcal M(\delta)$ and $\mathcal M_-(\delta)$ below are error
thresholds for the present low-rank denoising problem, not probabilities of
exact recovery.

Put
\begin{equation}
    k:=r\max(d_1,d_2),
    \qquad
    s_r:=r(d_1+d_2-r),
\label{eq:lowRankDimensions}
\end{equation}
so that $k\leq s_r\leq2k$. The number $s_r$ is the dimension of the
rank-$r$ stratum. Moreover, $k$ is the largest possible dimension of a
linear space consisting entirely of matrices of rank at most $r$, by the
bounded-rank subspace theorem of Flanders~\cite{Flanders1962}.

Let \(\kappa_j:=\pi^{j/2}/\Gamma(j/2+1)\), \(j\geq0\),
denote the volume of the Euclidean unit ball in $\mathbb R^j$, with
$\kappa_0=1$. For $\delta\in(0,1/2)$, define
$q_{D,\delta}\in(0,1)$ by
\begin{equation}
    \mathbb P(U_1\geq q_{D,\delta})=\delta,
\label{eq:ballCoordinateQuantile}
\end{equation}
where $U_1$ is any fixed coordinate of a vector uniformly distributed on
the unit ball in $\mathbb R^D$.

\begin{lemma}
\label{lem:ballVolumeAndCapBounds}
Uniformly over $1\leq k\leq D$,
\begin{equation}
\frac{k}{eD}
\leq
\left(
\frac{\kappa_D}{\kappa_k\kappa_{D-k}}
\right)^{2/k}
\leq
\frac{2k}{D}.
\label{eq:lowRankVolumeRatioBounds}
\end{equation}
Moreover, there are universal constants $c_1,C_1>0$ such that
\begin{equation}
c_1\min\left\{\frac{\log(1/\delta)}D,1\right\}
\leq
q_{D,\delta}^2
\leq
C_1\min\left\{\frac{\log(1/\delta)}D,1\right\},
\qquad
0<\delta\leq\frac14.
\label{eq:lowRankBallCoordinateRate}
\end{equation}
\end{lemma}

\begin{proof}
Put $a=k/2$, $b=(D-k)/2$, and $n=a+b$. Then
\(R_{a,b}:=\kappa_D/(\kappa_k\kappa_{D-k})
=\Gamma(a+1)\Gamma(b+1)/\Gamma(n+1)
=a\int_0^1t^{a-1}(1-t)^b\,\mathrm dt\).
Restricting the integral to $[0,a/n]$ and using
$y\log y\geq y-1$ gives
\(R_{a,b}\geq(a/n)^a(1-a/n)^b\geq(a/(en))^a\).
For the reverse bound, $R_{a,b}\leq1$ because
$B_D\subseteq B_k\times B_{D-k}$. This suffices when $a\geq b$. If
$b\geq a\geq1$, the standard Gamma-ratio bounds
\(\Gamma(n+1)/\Gamma(b+1)\geq b^a\) and
\(\Gamma(a+1)\leq a^a\) give
\(R_{a,b}^{1/a}\leq a/b\leq2a/n\).
The only remaining value is $a=1/2$; Wendel's inequality gives
\(\Gamma(b+3/2)/\Gamma(b+1)\geq(b+1/2)^{1/2}\),
and hence the same conclusion. Taking the power $1/a=2/k$ proves
\eqref{eq:lowRankVolumeRatioBounds}.

For the cap bound, write
\(T_D(t):=\mathbb P(U_1\geq t)
=c_D\int_t^1(1-x^2)^{(D-1)/2}\,\mathrm dx\), where
\(c_D:=\kappa_{D-1}/\kappa_D\asymp\sqrt D\).
Let $A=(D+1)/2$ and $G_D(t)=\tfrac12(1-t^2)^A$. The functions
$G_D$ and $T_D$ agree at $0$ and $1$, while
\((G_D-T_D)'(t)=(c_D-At)(1-t^2)^{A-1}\).
Since $c_D\leq A$, this derivative changes sign at most once, so
\(T_D(t)\leq G_D(t)\leq\frac12\exp(-At^2)\).
At $t=q_{D,\delta}$ this proves the upper bound in
\eqref{eq:lowRankBallCoordinateRate}.

For the lower bound, $\delta\leq1/4$ and the bound
$c_D\leq C\sqrt D$ imply
\(1/4\leq\mathbb P(0\leq U_1\leq q_{D,\delta})
\leq C\sqrt D\,q_{D,\delta}\),
so $q_{D,\delta}^2\geq c/D$. If $q_{D,\delta}\leq1/2$, integrate the
cap density over the interval of length $h=(8\sqrt D)^{-1}$ beginning at
$q_{D,\delta}$. Since $q_{D,\delta}+h\leq5/8$ and
$\log(1-x^2)\geq-2x^2$ on this interval,
\(\delta=T_D(q_{D,\delta})\geq
c\exp\{-CDq_{D,\delta}^2\}\).
Combining this with $q_{D,\delta}^2\geq c/D$ gives
\(q_{D,\delta}^2\geq c'\log(1/\delta)/D\).
If $q_{D,\delta}>1/2$, the claimed capped bound is immediate. This proves
\eqref{eq:lowRankBallCoordinateRate}.
\end{proof}

We are now ready to state the main finite-sample result for bounded-energy
low-rank denoising.

\begin{theorem}
\label{prop:boundedEnergyLowRankQuantile}
For every $\delta\in(0,1/2)$,
\begin{equation}
\mathcal M_-(\delta)
\geq
\frac{\tau_D^2}{D}
\max\left\{
(1-\delta)^{2/k}
\left(
\frac{\kappa_D}{\kappa_k\kappa_{D-k}}
\right)^{2/k},
q_{D,\delta}^2
\right\}.
\label{eq:boundedEnergyLowRankExactLower}
\end{equation}
Moreover, there are universal constants $0<c<C<\infty$ such that, for every
$\delta\in(0,1/4]$,
\begin{equation}
c\sigma^2
\min\left\{
\frac{s_r+\log(1/\delta)}{D},1
\right\}
\leq
\mathcal M_-(\delta)
\leq
\mathcal M(\delta)
\leq
C\sigma^2
\min\left\{
\frac{s_r+\log(2/\delta)}{D},1
\right\}.
\label{eq:boundedEnergyLowRankRate}
\end{equation}
Thus the full finite-sample quantile rate is determined by the intrinsic
low-rank dimension and the confidence level, up to universal constants.
    
\end{theorem}
Before going through the proof, let us comment on the result. The two terms in the lower bound represent different statistical
obstructions. The first lower bound captures the intrinsic dimensional cost of estimating
a rank-\(r\) matrix. At fixed confidence, it is of order
\(\sigma^2s_r/D\), where \(s_r=r(d_1+d_2-r)\) is the dimension of the
rank-\(r\) matrix model. The second lower bound captures the additional
cost of requiring failure probability at most \(\delta\), since
\(q_{D,\delta}^2\asymp
\min\{\log(1/\delta)/D,1\}\). Because
\(\max\{a,b\}\geq(a+b)/2\), the two obstructions combine to produce the
rate \(\sigma^2\min\{(s_r+\log(1/\delta))/D,1\}\).
The appearance of the low-rank dimension is consistent with classical
matrix-denoising theory under independent Gaussian noise, where minimax
mean-squared risks and the performance of singular-value thresholding have
been studied extensively; see, for example,
\cite{DonohoGavish2014}. The theorem does not claim a new low-rank
dimension scaling. Its application-specific novelty is instead the
finite-sample minimax-quantile characterisation for isotropic uniform
Frobenius-ball noise, including its explicit dependence on \(\delta\).
Existing work on uniform matrix-ball distributions focuses on their
probabilistic and spectral asymptotics \cite{KaufmannThale2022}, while work
on rotationally invariant matrix noise studies asymptotic spiked models and
algorithmic optimality \cite{DudejaLiuMa2026}. To our knowledge, neither
line of work gives the present two-sided finite-sample minimax-quantile
bound. The lower-bound argument is also qualitatively different from the
usual Gaussian packing proof: all nontrivial pairwise KL and R\'enyi
divergences are infinite, whereas  Maximal Leakage yields the correct lower bound through the exact Steiner
tube-volume formula.
\begin{proof}
We begin with the structural term, for which Maximal Leakage gives a
closed-form geometric converse. The rank-constrained parameter space \(\Theta_r\) contains a
\(k\)-dimensional linear subspace \(S\) of
\(\mathbb R^{d_1\times d_2}\). Specifically, if \(d_2\ge d_1\), let \(S\) consist of the matrices
supported on the first \(r\) rows; otherwise, let \(S\) consist of those
supported on the first \(r\) columns. For \(R>0\), let \(W\) be uniformly
distributed on \(B_S(R):=\{A\in S:\|A\|_F\leq R\}\),
with respect to the \(k\)-dimensional Lebesgue measure on \(S\).
Suppose $0<\rho\leq R^2/D$, and let \(\Pi_S\) denote the orthogonal projection onto \(S\).
For every \(\vartheta\in\Theta_r\) and \(w\in S\),
\(\|w-\Pi_S\vartheta\|_F\le \|w-\vartheta\|_F\).
Hence,
\(\{w\in B_S(R):\|w-\vartheta\|_F<\sqrt{D\rho}\}
\subseteq B_S(R)\cap B_S(\Pi_S\vartheta,\sqrt{D\rho})\).
The \(k\)-dimensional volume of the latter set is at most
\(\kappa_k {D\rho}^{k/2}\). Since \(\sqrt{D\rho}\le R\), this bound is attained by taking
\(\vartheta=0\). Therefore,
\begin{equation}
    L_W(\rho)
=
\left(\frac{\sqrt{D\rho}}{R}\right)^k.
    \label{eq:lowRankPriorSmallBall}
\end{equation}

Let
\( 
B_D:=\{A\in\mathbb R^{d_1\times d_2}:\|A\|_F\le1\}
\)
denote the unit Frobenius ball. Conditional on \(W=w\), the observation
\(Y\) is uniformly distributed on \(w+\tau_D B_D\), and therefore has
Lebesgue density
\(p_{Y\mid W=w}(y)=\mathbf 1\{y\in w+\tau_D B_D\}/
(\kappa_D\tau_D^D)\). Consequently,
\(\operatorname*{ess\,sup}_{w\in B_S(R)}p_{Y\mid W=w}(y)
=\mathbf 1\{y\in B_S(R)+\tau_D B_D\}/
(\kappa_D\tau_D^D)\),
up to a Lebesgue-null boundary set. Therefore~\eqref{eq:maximalLeakageIntegralDefinition} gives
\begin{equation}
\exp\{\ml{W}{Y}\}
=
\frac{
\operatorname{Vol}_D(B_S(R)+\tau_DB_D)
}{
\kappa_D\tau_D^D
}
=:
\Lambda_{D,k}(R/\tau_D).
\label{eq:lowRankLeakageTubeRatio}
\end{equation}
The Steiner formula for parallel bodies~\cite{Schneider2014} and the
intrinsic volumes of a $k$-ball yield the exact polynomial
\begin{equation}
\Lambda_{D,k}(t)
=
\sum_{j=0}^{k}
\binom{k}{j}
\frac{\kappa_k\kappa_{D-j}}
     {\kappa_{k-j}\kappa_D}
t^j.
\label{eq:lowRankSteinerPolynomial}
\end{equation}
Indeed, the Steiner expansion is
\(\operatorname{Vol}_D(K+uB_D)=
\sum_{j=0}^k\kappa_{D-j}V_j(K)u^{D-j}\)
for a $k$-dimensional convex body $K$, while
\(V_j(B_S(R))=\binom{k}{j}\kappa_kR^j/\kappa_{k-j}\).

Combining~\eqref{eq:lowRankPriorSmallBall},
\eqref{eq:lowRankLeakageTubeRatio}, and
\eqref{eq:maximalLeakageSuccessBound} gives, for every estimator
$\widehat M$ and every $0<\rho\leq R^2/D$,
\begin{equation}
P_{WY}\!\left(
    \ell_F(\widehat M(Y),W)<\rho
\right)
\leq
\left(
    \frac{\sqrt{D\rho}}{R}
\right)^k
\Lambda_{D,k}(R/\tau_D).
\label{eq:lowRankMaximalLeakageSuccess}
\end{equation}
Define
\(
\rho_\star
:=
\frac{R^2}{D}
\left\{
    \frac{1-\delta}
         {\Lambda_{D,k}(R/\tau_D)}
\right\}^{2/k}.
\)
At $\rho=\rho_\star$, the right-hand side of
\eqref{eq:lowRankMaximalLeakageSuccess} equals $1-\delta$.  To satisfy
the strict inequality required by
\eqref{eq:maximalLeakageSuccessBound}, fix $\varepsilon\in(0,1)$ and set
$\rho_\varepsilon=(1-\varepsilon)\rho_\star$.  Then
\begin{align}
\left(
    \frac{\sqrt{D\rho_\varepsilon}}{R}
\right)^k
\Lambda_{D,k}(R/\tau_D)
&=
(1-\varepsilon)^{k/2}(1-\delta)
\\
&<
1-\delta.
\end{align}
It follows that \(\mathcal M_-(\delta)\geq\rho_\varepsilon\) for every
\(\varepsilon\in(0,1)\).
Letting \(\varepsilon\downarrow0\) gives, for every \(R>0\),
\begin{equation}
\mathcal M_-(\delta)
\ge
\frac{R^2}{D}
\left\{
    \frac{1-\delta}
         {\Lambda_{D,k}(R/\tau_D)}
\right\}^{2/k}.
\label{eq:lowRankFiniteThresholdLower}
\end{equation}
Since \(\Lambda_{D,k}\) is a polynomial of degree \(k\) with leading
coefficient
\(
\frac{\kappa_k\kappa_{D-k}}{\kappa_D},
\)
we have
\(
\Lambda_{D,k}(R/\tau_D)
\sim
\frac{\kappa_k\kappa_{D-k}}{\kappa_D}
\left(\frac{R}{\tau_D}\right)^k\text{ as }R\to\infty.
\)
Therefore, letting \(R\to\infty\) in
\eqref{eq:lowRankFiniteThresholdLower} yields
\[
\mathcal M_-(\delta)\ge
\frac{\tau_D^2}{D}(1-\delta)^{2/k}
\left(\frac{\kappa_D}{\kappa_k\kappa_{D-k}}\right)^{2/k},
\]
which is the first term in
\eqref{eq:boundedEnergyLowRankExactLower}.

For the confidence term, fix a rank-one matrix $E$ with $\|E\|_F=1$ and,
for $q\in(0,1)$, define \(M_-:=-\tau_DqE\) and
\(M_+:=\tau_DqE\).
Let $P_-$ and $P_+$ denote the distributions of $Y$ under
$M=M_-$ and $M=M_+$, respectively.
Both $M_-$ and $M_+$ belong to $\Theta_r$.  Thus, any estimator that
performs uniformly well over $\Theta_r$ must also perform well when the
unknown signal matrix is known to be one of these two values.
Let $P_-$ and $P_+$ denote the distributions of the observation $Y$
when the true signal matrix is $M_-$ and $M_+$, respectively.
$P_-$ is uniform on the ball $M_-+\tau_DB_D$, while $P_+$ is
uniform on the ball $M_++\tau_DB_D$. Both supports are translates of $\tau_DB_D$,
and the Frobenius distance between their centres is
\(\|M_+-M_-\|_F=2\tau_Dq\).
Define
\(
U_1:=\langle U,E\rangle_F,
\)
the coordinate of the noise matrix in the direction $E$.  By rotational
invariance, $U_1$ has the same distribution as the first coordinate of a
random vector uniformly distributed on the unit ball in $\mathbb R^D$.

The intersection of the two balls
\(
M_-+\tau_DB_D
\text{ and }
M_++\tau_DB_D
\)
is a symmetric lens.  The hyperplane orthogonal to $E$ through the origin
divides this lens into two congruent spherical caps.  After translating
and rescaling either ball to the unit ball, each cap corresponds to the
event $\{U_1\geq q\}$.  Consequently,
\[
\frac{
    \operatorname{Vol}_D\!\left[
        (M_-+\tau_DB_D)\cap(M_++\tau_DB_D)
    \right]
}{
    \operatorname{Vol}_D(\tau_DB_D)
}
=
2\mathbb P(U_1\geq q).
\]
Because $P_-$ and $P_+$ are uniform distributions on sets of equal
volume,
\begin{equation} 
\frac{1-\operatorname{TV}(P_-,P_+)}{2}
=
\frac{
    \operatorname{Vol}_D\!\left[
        (M_-+\tau_DB_D)\cap(M_++\tau_DB_D)
    \right]
}{
    2\operatorname{Vol}_D(\tau_DB_D)
} = \mathbb P(U_1\geq q).
\label{eq:lowRankBallOverlapTV}
\end{equation}
Define $q_{D,\delta}\in(0,1)$ by
\begin{equation}
\mathbb P(U_E\geq q_{D,\delta})=\delta,
\qquad
U_E:=\langle U,E\rangle_F.
\label{eq:lowRankCapQuantileDefinition}
\end{equation}
Since the function
\(
q\longmapsto\mathbb P(U_E\geq q)
\)
is strictly decreasing, every $q<q_{D,\delta}$ satisfies
\(
\mathbb P(U_E\geq q)>\delta.
\)

Combining this inequality with
\eqref{eq:lowRankBallOverlapTV} gives
\((1-\operatorname{TV}(P_-^{(q)},P_+^{(q)}))/2>\delta\).
Therefore,~\Cref{cor:leCam} applies at the loss threshold
$\tau_D^2q^2/D$.  Letting $q\uparrow q_{D,\delta}$ yields
\[
\mathcal M_-(\delta)
\geq
\frac{\tau_D^2}{D}q_{D,\delta}^2,
\]
proving the second term in
\eqref{eq:boundedEnergyLowRankExactLower}.

We next simplify
\eqref{eq:boundedEnergyLowRankExactLower} to obtain the lower bound in
\eqref{eq:boundedEnergyLowRankRate}. For \(0<\delta\le1/4\), we have
\((1-\delta)^{2/k}\ge9/16\). Hence
\Cref{lem:ballVolumeAndCapBounds} and
\eqref{eq:boundedEnergyLowRankExactLower} give
\[
\mathcal M_-(\delta)
\gtrsim
\frac{\tau_D^2}{D}
\max\left\{
    \frac{k}{D},
    \min\left\{\frac{\log(1/\delta)}{D},1\right\}
\right\}.
\]
Using \(k\le s_r\le2k\) and
\(\tau_D^2/D=\sigma^2(D+2)/D\asymp\sigma^2\), we obtain the lower bound
in \eqref{eq:boundedEnergyLowRankRate}.

For the upper bound, choose a measurable best rank-\(r\) approximation
of \(Y\) in Frobenius norm:
\(
\widehat M_r
\in
\operatorname*{arg\,min}_{\operatorname{rank}(A)\le r}
\|Y-A\|_F.
\)
Such a selection exists by choosing a Borel-measurable version of the
truncated singular value decomposition. For a matrix \(A\) and an integer
\(m\ge1\), define \(Z_m(A):=\sup\{\langle A,H\rangle_F:
\operatorname{rank}(H)\le m,\ \|H\|_F=1\}\).
Equivalently,
\begin{equation}
Z_m(A)
=
\left(
    \sum_{j=1}^{m\wedge d_1\wedge d_2}
    s_j(A)^2
\right)^{1/2},\label{eq:zmVariational}
\end{equation}
where \(s_1(A)\ge s_2(A)\ge\cdots\) are the singular values of \(A\).
Let \(\Delta:=\widehat M_r-M\). Since \(M\in\Theta_r\), the optimality
of \(\widehat M_r\) gives
\(
\|Y-\widehat M_r\|_F^2
\le
\|Y-M\|_F^2.
\)
Substituting \(Y=M+\tau_DU\) and expanding yields
\(
\|\Delta\|_F^2
\le
2\tau_D\langle U,\Delta\rangle_F.
\)
Moreover, \(\operatorname{rank}(\Delta)\le2r\), and hence
\(\langle U,\Delta\rangle_F\le Z_{2r}(U)\|\Delta\|_F\).
Combining the preceding inequalities, we obtain either \(\Delta=0\),
in which case the desired bound is immediate, or the following bound,
which therefore holds in all cases:
\begin{equation}
\|\widehat M_r-M\|_F
\le
2\tau_D Z_{2r}(U).
\label{eq:lowRankProjectionBasicInequality}
\end{equation}
Let \(G\in\mathbb R^{d_1\times d_2}\) have independent standard Gaussian
entries, and let \(T\) be independent of \(G\) with distribution
\(
\mathbb P(T\le t)=t^D,
\text{ } 0\le t\le1.
\)
Then
\(
U\stackrel{d}{=}T\frac{G}{\|G\|_F}.
\) The functional \(Z_{2r}\) is one-Lipschitz with respect to the Frobenius
norm. Indeed, for any matrices \(A\) and \(B\),
\[
|Z_{2r}(A)-Z_{2r}(B)|
\le
\sup_{\substack{\operatorname{rank}(H)\le2r\\
                 \|H\|_F=1}}
|\langle A-B,H\rangle_F|
\le
\|A-B\|_F.
\]
Let \(s_1(G)\ge s_2(G)\ge\cdots\) denote the singular values of \(G\). 
Since \(s_j(G)\le s_1(G)=\|G\|_{\mathrm{op}}\), it follows from~\Cref{eq:zmVariational} that
\(Z_{2r}(G)\le\sqrt{2r}\,\|G\|_{\mathrm{op}}\).
For a \(d_1\times d_2\) matrix with independent standard Gaussian
entries, the standard operator-norm bound \cite[Proposition~10.1]{HalkoMartinssonTropp2011} gives
\(\mathbb E\|G\|_{\mathrm{op}}\le\sqrt{d_1}+\sqrt{d_2}\).
Consequently, \(\mathbb E Z_{2r}(G)\le
\sqrt{2r}(\sqrt{d_1}+\sqrt{d_2})\). Set
\(
x:=\log(2/\delta).
\)
Since \(Z_{2r}\) is one-Lipschitz, Gaussian concentration
\cite[Proposition~10.3]{HalkoMartinssonTropp2011} gives
\(Z_{2r}(G)\le\mathbb E Z_{2r}(G)+\sqrt{2x}\)
with probability at least \(1-\delta/2\). Moreover, since
\(\|G\|_F^2\sim\chi_D^2\), the lower chi-square tail bound
\cite[Lemma~1]{LaurentMassart2000} gives
\(
\|G\|_F^2
\ge
D-2\sqrt{Dx}
\)
with probability at least \(1-\delta/2\). Under the assumption
\(x\le D/16\), the latter bound implies
\(
\|G\|_F\ge\sqrt{D/2}.
\)

On the intersection of these two events, \(T\le1\) and the positive
homogeneity of \(Z_{2r}\) give
\(Z_{2r}(U)\le Z_{2r}(G)/\|G\|_F\). Using the preceding expectation
bound and \(\sqrt{d_1}+\sqrt{d_2}\le\sqrt{2(d_1+d_2)}\), a union bound
therefore yields, with probability at least \(1-\delta\),
\begin{equation}
Z_{2r}(U)
\le
\frac{C_2}{\sqrt D}
\left\{
\sqrt{r(d_1+d_2)}
+
\sqrt{\log(2/\delta)}
\right\},
\label{eq:lowRankUniformBallWidth}
\end{equation}
where \(C_2>0\) is a universal constant. Combining
\eqref{eq:lowRankProjectionBasicInequality} and
\eqref{eq:lowRankUniformBallWidth}, and using
\((a+b)^2\le2(a^2+b^2)\), we obtain, with probability at least
\(1-\delta\),
\[
\ell_F(\widehat M_r,M)
\le
\frac{8C_2^2\tau_D^2}{D^2}
\left\{
r(d_1+d_2)+\log(2/\delta)
\right\}
\le
C_3\sigma^2
\frac{s_r+\log(2/\delta)}{D},
\]
where we used \(r(d_1+d_2)\le2s_r\) and
\(\tau_D^2/D=\sigma^2(D+2)/D\le3\sigma^2\). Since this bound is uniform
over \(M\in\Theta_r\), it gives the required upper quantile bound when
\(\log(2/\delta)\le D/16\). It remains to consider \(\log(2/\delta)>D/16\). Since \(M\) is a
feasible rank-\(r\) approximation, the optimality of \(\widehat M_r\)
and the triangle inequality give
\[
\|\widehat M_r-M\|_F
\le
\|\widehat M_r-Y\|_F+\|Y-M\|_F
\le
2\|Y-M\|_F
\le
2\tau_D.
\]
Consequently, \(\ell_F(\widehat M_r,M)\le4\tau_D^2/D\le12\sigma^2\).
Since \(s_r \geq 0\) and we are now considering
\(\log(2/\delta)>D/16\), we have
\(\min\{(s_r+\log(2/\delta))/D,1\}\ge1/16\).
Since the minimum above is at least \(1/16\), the deterministic bound
\(\ell_F(\widehat M_r,M)\le12\sigma^2\) yields
\[
\ell_F(\widehat M_r,M)
\le
C\sigma^2
\min\left\{
\frac{s_r+\log(2/\delta)}{D},1
\right\}
\]
for a sufficiently large \(C\).
Together with the preceding high-probability bound, this proves
\eqref{eq:boundedEnergyLowRankRate}.
\end{proof}
The role of Maximal Leakage is particularly transparent in this model,
because the usual KL-based Fano argument is vacuous, as explained next.
\begin{remark}
\label{rmk:lowRankNonKL}
For distinct \(M,M'\), each translated noise ball contains a set of
positive \(D\)-dimensional volume that lies outside the other.
Consequently, \eqref{eq:lowRankPairwiseDivergencesInfinite} holds, and
the usual pairwise-divergence packing argument is vacuous. The prior-based information quantities remain finite. Indeed,
\(I(W;Y)\le\ml{W}{Y}=\log\Lambda_{D,k}(R/\tau_D)<\infty\).
By the monotonicity of Sibson mutual information in its order, the same is true
of \(I_\alpha(W;Y)\) for every finite \(\alpha>1\). More precisely, as
\(R\to\infty\),
\begin{equation}
I(W;Y)=k\log R+O(1),
\label{eq:lowRankShannonGrowth}
\end{equation}
where \(D,k\), and \(\tau_D\) are fixed. The upper bound follows because \(\Lambda_{D,k}\) is a polynomial of
degree \(k\), and hence
\(\log\Lambda_{D,k}(R/\tau_D)=k\log R+O(1)\). For the reverse bound, set
\(Z_S:=\tau_D\Pi_SU\). Since \(\Pi_SY=W+Z_S\), data processing gives
\(I(W;Y)\ge I(W;W+Z_S)\).
Since \(W\) and \(Z_S\) are independent, the additive-noise identity
gives \(I(W;W+Z_S)=h(W+Z_S)-h(Z_S)\).
The entropy-power inequality implies \(h(W+Z_S)\ge h(W)\). Therefore,
\(I(W;W+Z_S)\ge h(W)-h(Z_S)=k\log R+O(1)\),
because \(h(W)=k\log R+\log\kappa_k\), whereas \(h(Z_S)\) is finite and
independent of \(R\), where \(h\) denotes the \(k\)-dimensional differential entropy. On the other hand, since
\(L_W(\rho)=(\sqrt{D\rho}/R)^k\), the binary KL divergence satisfies
\begin{equation}
d_{\mathrm{KL}}
\left(1-\delta\middle\|L_W(\rho)\right)
=
(1-\delta)k\log R+O(1)
\label{eq:lowRankBinaryKLGrowth}
\end{equation}
as \(R\to\infty\), for fixed \(\rho>0\) and
\(\delta\in(0,1)\). Consequently,
\[I(W;Y)
-
d_{\mathrm{KL}}
\left(1-\delta\middle\|L_W(\rho)\right)
=
\delta k\log R+O(1)
\longrightarrow\infty
\qquad (R\to\infty).\]
Thus, for every fixed \(\delta>0\), the binary-KL sufficient condition
\(I(W;Y)<d_{\mathrm{KL}}(1-\delta\|L_W(\rho))\)
fails for all sufficiently large \(R\). By contrast, the Maximal Leakage bound depends on the product
\(L_W(\rho)\exp\{\ml{W}{Y}\}\). Since
\[
L_W(\rho)
=
\left(\frac{\sqrt{D\rho}}{R}\right)^k
\quad\text{and}\quad
\exp\{\ml{W}{Y}\}
=
\Lambda_{D,k}(R/\tau_D)
\sim
\frac{\kappa_k\kappa_{D-k}}{\kappa_D}
\left(\frac{R}{\tau_D}\right)^k,
\]
the leading powers of \(R^k\) cancel, and
\[
L_W(\rho)\exp\{\ml{W}{Y}\}
\longrightarrow
\frac{\kappa_k\kappa_{D-k}}{\kappa_D}
\left(\frac{\sqrt{D\rho}}{\tau_D}\right)^k.
\]
Thus the Maximal Leakage converse remains nondegenerate in the
diffuse-prior limit. Total-variation arguments are not ruled out, but
Maximal Leakage provides a particularly tractable global converse here,
reducing the information term to the volume of a parallel body.
\end{remark}

To the best of our knowledge, the two-sided
finite-sample minimax-quantile rate bounds above for uniform Frobenius-ball
noise, and their exact Maximal Leakage lower-bound, have not appeared in the literature. 

\subsection{Approximate recovery: finite-order Sibson mutual information}
\label{sec:bsc-hamming}

Maximal Leakage is naturally suited to exact recovery, where success
requires identifying a single parameter. We now consider approximate
recovery under Hamming loss, for which a finite Sibson order can give a
strictly stronger converse.

Let \(W=(W_1,\ldots,W_d)\) be uniform on \(\{-1,+1\}^d\). Conditional on
\(W=w\), we observe \(X_i=w_iZ_i\), where \(Z_1,\ldots,Z_d\) are
independent and
\[
P(Z_i=-1)=p_{i,d},
\qquad
P(Z_i=1)=1-p_{i,d},
\qquad
0<p_{i,d}<\frac12.
\]
We use the Hamming loss
\(d_{\mathrm H}(\widehat w,w)
:=\sum_{i=1}^d\1_{\{\widehat w_i\neq w_i\}}\).
For \(r\in\{0,\ldots,d\}\), define
\[
S_d^\star(r)
:=
\sup_{\widehat W}
\inf_{w\in\{-1,+1\}^d}
P_{X\mid w}
\left\{
d_{\mathrm H}(\widehat W(X),w)\le r
\right\}.
\]
Let \(V_d(r):=\sum_{j=0}^r\binom dj\), and, for \(\alpha>1\), let
\(H_\alpha(p):=(1-\alpha)^{-1}
\log\{p^\alpha+(1-p)^\alpha\}\).

\begin{proposition}
\label{prop:heterogeneousSibsonConverse}
Let \(B_{i,d}\sim\operatorname{Bernoulli}(p_{i,d})\) be independent.
Then, for every \(r\in\{0,\ldots,d\}\) and \(\alpha>1\),
\begin{equation}
S_d^\star(r)
=
P\left\{\sum_{i=1}^dB_{i,d}\le r\right\}
\le
\left[
V_d(r)
\exp\left\{
-\sum_{i=1}^dH_\alpha(p_{i,d})
\right\}
\right]^{\frac{\alpha-1}{\alpha}}
\wedge1.
\label{eq:heterogeneousSibsonBound}
\end{equation}
Consequently, if the right-hand side is smaller than \(1-\delta\), then
\(\mathcal M_-(\delta)\ge r+1\).
\end{proposition}

\begin{proof}
Conditional on \(X=x\), the coordinates of \(W\) are independent and
\(P(W_i\neq x_i\mid X=x)=p_{i,d}\). For a candidate decision
\(\widehat w\), the posterior error indicators
\(\1_{\{\widehat w_i\neq W_i\}}\) are therefore independent Bernoulli
variables with parameters
\[
p_{i,d}
\quad\text{if }\widehat w_i=x_i,
\qquad
1-p_{i,d}
\quad\text{if }\widehat w_i=-x_i.
\]
Because \(p_{i,d}<1/2\), each parameter is minimised by choosing
\(\widehat w_i=x_i\). The resulting sum of error indicators is
stochastically no larger than for any other decision. Hence
\(\widehat W(x)=x\) maximises the posterior probability of every Hamming
ball.

For \(\widehat W(X)=X\), the error indicators are
\(\1_{\{Z_i=-1\}}\), so their joint distribution does not depend on
\(w\). Let \(s_r\) denote this estimator's probability of success within
Hamming distance \(r\). Since the estimator is Bayes optimal under the uniform
prior, every estimator \(\widetilde W\) satisfies
\[
\inf_w
P_{X\mid w}\!\left\{
d_{\mathrm H}(\widetilde W(X),w)\le r
\right\}
\le
2^{-d}\sum_w
P_{X\mid w}\!\left\{
d_{\mathrm H}(\widetilde W(X),w)\le r
\right\}
\le s_r.
\]
The estimator \(\widehat W(X)=X\) attains \(s_r\) for every \(w\).
It is therefore minimax, which proves the equality in
\eqref{eq:heterogeneousSibsonBound}.

For the uniform prior,
\(L_W(r+1)=V_d(r)/2^d\). Moreover,
\[
I_\alpha(W_i,X_i)
=
\frac{\alpha}{\alpha-1}
\log\left[
2\left\{
\frac{p_{i,d}^\alpha+(1-p_{i,d})^\alpha}{2}
\right\}^{1/\alpha}
\right]
=
\log2-H_\alpha(p_{i,d}).
\]
Substitution into the H\"older relaxation in terms of \(I_\alpha\) proves the
inequality.
\end{proof}

We next compare the exponent supplied by
\Cref{prop:heterogeneousSibsonConverse} with the exact success exponent.
Suppose that
\[
\nu_d:=\frac1d\sum_{i=1}^d\delta_{p_{i,d}}
\Longrightarrow\nu,
\]
where \(\nu\) is supported on a compact subset of \((0,1/2)\). We now specialise the nonasymptotic bound to proportional Hamming
recovery. Fix \(0<\rho<\int p\,\nu(\mathrm dp)\) and, for each \(d\),
set \(r=r_d:=\lfloor\rho d\rfloor\).
Define the exact lower-tail exponent
\begin{equation}
I_\nu(\rho)
:=
\sup_{s\le0}
\left\{
s\rho-
\int\log(1-p+pe^s)\,\nu(\mathrm dp)
\right\},
\label{eq:heterogeneousExactExponent}
\end{equation}
and the order-\(\alpha\) Sibson exponent
\begin{equation}
J_\nu(\rho,\alpha)
:=
\frac{\alpha-1}{\alpha}
\left\{
\int H_\alpha(p)\,\nu(\mathrm dp)-h_2(\rho)
\right\}.
\label{eq:heterogeneousSibsonExponent}
\end{equation}
Only $J_\nu(\rho,\alpha)>0$ yields a nontrivial exponential upper bound
on the success probability.

The next proposition shows that, under the stated entropy condition,
optimising over a finite Sibson order yields a strictly larger decay
exponent than Maximal Leakage, while remaining below the exact minimax
success exponent.

\begin{proposition}
\label{prop:finiteSibsonBeatsLeakage}
The exact minimax success exponent satisfies
\begin{equation}
\lim_{d\to\infty}
-\frac1d\log S_d^\star(r_d)
=
I_\nu(\rho).
\label{eq:heterogeneousExactLD}
\end{equation}
If
\[
h_2(\rho)
<
\int h_2(p)\,\nu(\mathrm dp),
\]
then \(J_\nu(\rho,\alpha)\) has a unique maximiser
\(\alpha^\star\in(1,\infty)\), determined by
\begin{equation}
\int
h_2\left(
\frac{p^{\alpha^\star}}
     {p^{\alpha^\star}+(1-p)^{\alpha^\star}}
\right)
\nu(\mathrm dp)
=
h_2(\rho).
\label{eq:heterogeneousOptimalOrder}
\end{equation}
Moreover,
\begin{equation}
0
<
J_\nu^\star(\rho)
\le
I_\nu(\rho),
\qquad
J_\nu^\star(\rho)
>
[J_\nu(\rho,\infty)]_+,
\label{eq:heterogeneousFiniteOrderBeatsLeakage}
\end{equation}
where
\(J_\nu(\rho,\infty)
=\int-\log(1-p)\,\nu(\mathrm dp)-h_2(\rho)\).
\end{proposition}

\begin{proof}
By independence of the \(B_{i,d}\),
\[
\frac1d\log
\mathbb E\exp\left\{
s\sum_{i=1}^dB_{i,d}
\right\}
=
\frac1d\sum_{i=1}^d
\log(1-p_{i,d}+p_{i,d}e^s)
=
\int\log(1-p+pe^s)\,\nu_d(\mathrm dp).
\]
For each fixed \(s\in\mathbb R\), the integrand is continuous on the
support of \(\nu\). Hence \(\nu_d\Rightarrow\nu\) implies
\[
\frac1d\log
\mathbb E\exp\left\{
s\sum_{i=1}^dB_{i,d}
\right\}
\longrightarrow
\int\log(1-p+pe^s)\,\nu(\mathrm dp).
\]
Let \(S_d:=\sum_{i=1}^dB_{i,d}\) and define
\[
\Lambda(s)
:=
\int\log(1-p+pe^s)\,\nu(\mathrm dp).
\]
The preceding calculation shows that
\(d^{-1}\log\mathbb E e^{sS_d}\to\Lambda(s)\) for every
\(s\in\mathbb R\). The function \(\Lambda\) is finite and continuously
differentiable on \(\mathbb R\), with
\[
\Lambda'(s)
=
\int
\frac{pe^s}{1-p+pe^s}\,
\nu(\mathrm dp).
\]
Since \(S_d/d\in[0,1]\), exponential tightness is automatic. The
G\"artner--Ellis theorem
\cite[Theorem~2.3.6]{DemboZeitouni2010} therefore implies that
\((S_d/d)\) satisfies a large-deviation principle with speed \(d\) and
rate function
\[
\Lambda^\star(x)
:=
\sup_{s\in\mathbb R}\{sx-\Lambda(s)\}.
\]

Now \(\Lambda'(0)=\int p\,\nu(\mathrm dp)>\rho\). Hence the minimiser of
\(\Lambda^\star\) over \((-\infty,\rho]\) is \(x=\rho\), and the
maximiser in the definition of \(\Lambda^\star(\rho)\) is nonpositive.
Therefore,
\[
\inf_{x\le\rho}\Lambda^\star(x)
=
\Lambda^\star(\rho)
=
\sup_{s\le0}
\left\{
s\rho-
\int\log(1-p+pe^s)\,\nu(\mathrm dp)
\right\}
=
I_\nu(\rho).
\]
Since \(r_d/d\to\rho\) and \(\Lambda^\star\) is continuous at \(\rho\),
the large-deviation upper and lower bounds give
\[
\lim_{d\to\infty}
-\frac1d
\log
P\left\{
S_d\le r_d
\right\}
=
I_\nu(\rho).
\]
Using
\(S_d^\star(r_d)=P\{S_d\le r_d\}\) proves
\eqref{eq:heterogeneousExactLD}.

Since \(\rho<1/2\), the Hamming-ball volume satisfies
\(d^{-1}\log V_d(r_d)\to h_2(\rho)\)
\cite[Chapter~11]{CoverThomas2006}. Moreover,
\(\nu_d\Rightarrow\nu\) implies
\(d^{-1}\sum_{i=1}^dH_\alpha(p_{i,d})
\to\int H_\alpha(p)\,\nu(\mathrm dp)\).
Taking negative logarithms in
\eqref{eq:heterogeneousSibsonBound} therefore gives
\[
\liminf_{d\to\infty}
-\frac1d\log S_d^\star(r_d)
\ge
[J_\nu(\rho,\alpha)]_+.
\]

To optimise this bound over \(\alpha\), set
\(q_\alpha(p):=
p^\alpha/\{p^\alpha+(1-p)^\alpha\}\). Since
\[
\frac{\alpha-1}{\alpha}H_\alpha(p)
=
-\frac1\alpha
\log\{p^\alpha+(1-p)^\alpha\},
\]
direct differentiation gives
\[
\frac{\partial}{\partial\alpha}
\left\{
\frac{\alpha-1}{\alpha}H_\alpha(p)
\right\}
=
\frac1{\alpha^2}h_2(q_\alpha(p)).
\]
Since \(0\le h_2(q_\alpha(p))\le\log2\), the derivative of the
integrand is uniformly bounded on every compact subinterval of
\((1,\infty)\). Differentiation under the integral \eqref{eq:heterogeneousSibsonExponent}  is therefore
justified by dominated convergence and yields
\begin{equation}
\frac{\partial}{\partial\alpha}J_\nu(\rho,\alpha)
=
\frac1{\alpha^2}
\left\{
\int h_2(q_\alpha(p))\,\nu(\mathrm dp)-h_2(\rho)
\right\}.
\label{eq:heterogeneousDerivative}
\end{equation}
For every \(p\in(0,1/2)\), the function
\(\alpha\mapsto q_\alpha(p)\) is continuous and strictly decreasing,
with \(q_1(p)=p\) and \(q_\alpha(p)\to0\) as
\(\alpha\to\infty\). Since \(h_2\) is strictly increasing on
\((0,1/2)\), the function
\[
\alpha\longmapsto
\int h_2(q_\alpha(p))\,\nu(\mathrm dp)
\]
is continuous and strictly decreasing from
\(\int h_2(p)\,\nu(\mathrm dp)\) to zero.

The assumption
\(h_2(\rho)<\int h_2(p)\,\nu(\mathrm dp)\), together with
\(h_2(\rho)>0\), therefore implies that there is a unique
\(\alpha^\star>1\) satisfying
\eqref{eq:heterogeneousOptimalOrder}. By
\eqref{eq:heterogeneousDerivative}, the derivative of
\(J_\nu(\rho,\alpha)\) is positive for
\(\alpha<\alpha^\star\) and negative for
\(\alpha>\alpha^\star\). Hence \(\alpha^\star\) is the unique maximiser.
Finally, since \(J_\nu(\rho,\alpha)\to0\) as \(\alpha\downarrow1\) and
the function initially increases, \(J_\nu^\star(\rho)>0\).
As \(\alpha\to\infty\),
\(H_\alpha(p)\to-\log(1-p)\). Dominated convergence therefore gives
\[
J_\nu(\rho,\infty)
=
\int-\log(1-p)\,\nu(\mathrm dp)-h_2(\rho),
\]
which is the untruncated exponent supplied by Maximal Leakage.

Since \(J_\nu(\rho,\alpha)\) is strictly decreasing for
\(\alpha>\alpha^\star\),
\[
J_\nu^\star(\rho)
=
J_\nu(\rho,\alpha^\star)
>
J_\nu(\rho,\infty).
\]
If \(J_\nu(\rho,\infty)\le0\), its effective exponent is zero, whereas
\(J_\nu^\star(\rho)>0\). Hence, in all cases,
\[
J_\nu^\star(\rho)
>
[J_\nu(\rho,\infty)]_+.
\]

Finally, \eqref{eq:heterogeneousSibsonBound} implies, for every
\(\alpha>1\),
\[
\liminf_{d\to\infty}
-\frac1d\log S_d^\star(r_d)
\ge
[J_\nu(\rho,\alpha)]_+.
\]
The left-hand side equals \(I_\nu(\rho)\) by
\eqref{eq:heterogeneousExactLD}. Taking the supremum over \(\alpha>1\)
and using \(J_\nu^\star(\rho)>0\) gives
\[
J_\nu^\star(\rho)\le I_\nu(\rho).
\]
\end{proof}
The preceding proposition gives a general condition under which a
finite Sibson order improves on Maximal Leakage. The following example
shows that the improvement can be qualitative: the Maximal Leakage
bound is trivial on the exponential scale, while the optimised
finite-order bound yields a strictly positive exponent. All numerical values in the following example are obtained by
one-dimensional root finding and optimisation, using natural
logarithms.

\begin{example}[Strict separation from Maximal Leakage]
\label{ex:heterogeneousSibsonSeparation}
Take
\[
\nu=\frac12\delta_{0.05}+\frac12\delta_{0.25},
\qquad
\rho=0.05.
\]
Then
\[
0
=
[J_\nu(0.05,\infty)]_+
<
J_\nu^\star(0.05) \simeq0.0439
<
I_\nu(0.05) \simeq0.0539.
\]
Maximal Leakage therefore gives only the trivial bound
\(S_d^\star(r_d)\le1\), whereas the optimised finite-order converse gives
\[
S_d^\star(r_d)
\le
\exp\{-(0.0439+o(1))d\}.
\]
Consequently, \(S_d^\star(r_d)\to0\) exponentially fast in $d$. Equivalently, for every estimator there exists a parameter \(w\) such that
\[
P_{X\mid w}\!\left\{d_{\mathrm H}(\widehat W,w)>r_d\right\}\longrightarrow1.
\]
In particular, for every fixed \(\delta\in(0,1)\), one has \(\mathcal M_-(\delta)\ge r_d+1\) for all sufficiently large \(d\), while Maximal Leakage yields only a trivial converse.
\end{example}

This section illustrates why the appropriate information measure depends on the resolution of the recovery problem. We consider a heterogeneous binary channel for which the exact minimax probability of recovery within Hamming distance \(r\) is available, and compare it with the converse obtained from Sibson mutual information of order \(\alpha\).

For exact recovery, the success set is a singleton, so Maximal Leakage naturally captures the problem. For approximate recovery with \(r_d\sim\rho d\), however, the success set is a Hamming ball containing exponentially many points. The converse then reflects a tradeoff between the volume of this ball and the likelihood-ratio moments controlled by \(I_\alpha\). Optimising over finite \(\alpha\) can exploit this tradeoff more effectively than Maximal Leakage, which in the example actually yields a trivial converse.
\subsection{Beyond classical information measures: Amemiya norms}
\label{sec:amemiya-outlier}

The preceding section showed that finite-order Sibson mutual information can be
better adapted than Maximal Leakage to approximate recovery. We now turn to
a different issue: the likelihood ratio may have an unbounded, non-polynomial
tail that is poorly represented by any fixed moment. Amemiya norms
allow the converse to be matched directly to this tail. Their role is
especially clear in one-coordinate localisation, where the success
probability is governed by the maximum of the coordinatewise likelihood
ratios and the inverse Young function determines the cost of searching over
the possible locations.

Let \(P\ll Q\) be probability measures on a measurable space
\(\mathsf X\), and write \(r:=\dd P/\dd Q\). For \(j\in[M]\), let
\(P_j:=Q^{\otimes(j-1)}\otimes P\otimes Q^{\otimes(M-j)}\). Thus, under
\(P_j\), the \(j\)-th coordinate is anomalous and all other coordinates have
law \(Q\). This is the standard one-outlier model; see, for example,
\cite{LiNitinawaratVeeravalli2014}. We allow an independent randomiser for
tie-breaking. This only enlarges the class of estimators, so every converse
below also applies to deterministic estimators. Define
\begin{equation}
S_M^\star
:=
\sup_{\widehat W}\min_{j\in[M]}P_j\{\widehat W=j\}.
\label{eq:oneCoordinateSuccess}
\end{equation}

Let \(\psi\) be a Young function and let \(\psi^\star\) be its complementary
Young function. We use the Amemiya norm
\(\|\cdot\|_{L_\psi^{\mathrm A}(Q)}\) defined above. The following lemma
isolates the general mechanism. The identity expresses the minimax
localisation success as the normalised expected maximum of the coordinatewise
likelihood ratios. The subsequent inequality controls this maximum through
an Amemiya norm adapted to the tail of \(r\), with
\(\psi^{-1}(M)\) capturing the cost of searching over \(M\) coordinates.

\begin{lemma}
\label{lem:oneCoordinateAmemiya}
For every Young function \(\psi\),
\begin{equation}
S_M^\star
=
\frac1M\Ee_{Q^{\otimes M}}
\max_{1\le j\le M}r(X_j)
\le
\frac{\|r\|_{L_\psi^{\mathrm A}(Q)}}
     {(\psi^\star)^{-1}(M)}
\le
\|r\|_{L_\psi^{\mathrm A}(Q)}
\frac{\psi^{-1}(M)}{M}.
\label{eq:oneCoordinateAmemiya}
\end{equation}
\end{lemma}

\begin{proof}
Give \(W\) the uniform distribution on \([M]\). Relative to
\(Q^{\otimes M}\), the likelihood under \(W=j\) is \(r(X_j)\), so a MAP rule
selects a coordinate maximising \(r(X_j)\). Its average success probability
is the first expression in \eqref{eq:oneCoordinateAmemiya}. With uniform
randomised tie-breaking, this rule is permutation invariant and has the same
success probability under every \(P_j\). Conversely, the minimum success
probability of any rule is at most its average under the uniform prior.
Hence the MAP rule is an equaliser Bayes rule and is minimax, proving the
identity.

For the success event \(E:=\{\widehat W=W\}\), the reference probability is
\((P_WQ^{\otimes M})(E)=1/M\), and the corresponding likelihood ratio is
\(r(X_W)\). The Amemiya--Luxemburg H\"older inequality therefore gives
\[
P_{WX}(E)
\le
\frac{\|r\|_{L_\psi^{\mathrm A}(Q)}}
     {(\psi^\star)^{-1}(M)}.
\]
Equivalently, this is
\Cref{cor:amemiyaHolder} with small-ball probability \(1/M\). Finally,
complementary Young functions satisfy
\(u\le\psi^{-1}(u)(\psi^\star)^{-1}(u)\le2u\); see
\cite[Chapter~2]{RaoRen1991}. The lower inequality in this relation yields
the last bound in \eqref{eq:oneCoordinateAmemiya}.
\end{proof}

We next give an example in which a Young function chosen to reflect large
likelihood-ratio values yields the correct localisation rate, whereas
fixed-power and untailored exponential choices give weaker bounds. Let
\(N_1,\ldots,N_M\) be independent \(\Pois(\lambda)\) variables, where
\(\lambda>0\), let \(W\) be uniform on \([M]\) and independent of them, and
observe
\begin{equation}
X_j=N_j+\1\{j=W\},
\qquad j\in[M].
\label{eq:taggedPoissonModel}
\end{equation}
Thus one tagged arrival is added to an unknown stream. The anomalous law is
the size-biased Poisson law \(1+\Pois(\lambda)\); see
\cite[Section~2]{ArratiaGoldsteinKochman2019}. With
\(Q=\Pois(\lambda)\), its likelihood ratio is \(r(k)=k/\lambda\).

Set \(h(u):=(1+u)\log(1+u)-u\) and define the Bennett-type Young function
\(\psi_{\mathrm B}(u):=\exp\{h(u)\}-1\). Its complementary Young function
is denoted by \(\psi_{\mathrm B}^\star\). Such functions underlie
Bennett-type maximal inequalities \cite{Wellner2017}. For
\(N\sim Q=\Pois(\lambda)\), set
\begin{equation}
K_\lambda
:=
\left\|\frac{N}{\lambda}\right\|_
{L_{\psi_{\mathrm B}}^{\mathrm A}(Q)}
=
\inf_{t>0}
\frac{1+\Ee_Q\psi_{\mathrm B}(tN/\lambda)}{t}.
\label{eq:taggedPoissonAmemiyaNorm}
\end{equation}

The following is the main result of this section. It shows that one fixed
Young function, chosen to reflect the distribution of the likelihood ratio,
reproduces the exact localisation rate.

\begin{proposition}
\label{prop:taggedPoisson}
For every fixed \(\lambda>0\), \(K_\lambda<\infty\) and
\begin{equation}
S_{M,\lambda}^\star
=
\frac{1}{M\lambda}\Ee\max_{1\le j\le M}N_j
\sim
\frac{\log M}{\lambda M\log\log M}.
\label{eq:taggedPoissonExactRate}
\end{equation}
Moreover,
\begin{equation}
S_{M,\lambda}^\star
\le
\frac{K_\lambda}{(\psi_{\mathrm B}^\star)^{-1}(M)}
\le
K_\lambda\frac{\psi_{\mathrm B}^{-1}(M)}{M}
\asymp_\lambda
\frac{\log M}{M\log\log M}.
\label{eq:taggedPoissonAmemiyaRate}
\end{equation}
Hence the Bennett-type Amemiya converse is rate-sharp. Consequently,
\[
\inf_{\widehat W}\max_{j\in[M]}P_j\{\widehat W\neq j\}
=
1-S_{M,\lambda}^\star
=
1-\frac{\log M}{\lambda M\log\log M}\{1+o(1)\}
\longrightarrow1.
\]
In particular, under zero--one loss, for every fixed
\(\delta\in(0,1)\), one has
\(\cM(\delta)=\cM_-(\delta)=1\) for all sufficiently large \(M\).
\end{proposition}

\begin{proof}
Since \(r(k)=k/\lambda\), the identity in
\eqref{eq:taggedPoissonExactRate} follows directly from
\Cref{lem:oneCoordinateAmemiya}. The Poisson maximum theorem gives
\[
\frac{\log\log M}{\log M}\max_{1\le j\le M}N_j
\longrightarrow1
\qquad\text{in probability};
\]
see \cite{Kimber1983}. The same convergence holds in mean. Indeed, for
every integer \(k\ge\lambda\), the Chernoff bound obtained from
\(\Ee e^{tN}=\exp\{\lambda(e^t-1)\}\) is
\(Q\{N\ge k\}\le(e\lambda/k)^k\). For every fixed
\(\varepsilon>0\), the union bound and the tail-sum formula give
\begin{align*}
&
\Ee\left[
\max_{1\le j\le M}N_j
-
\left\lceil
(1+\varepsilon)\frac{\log M}{\log\log M}
\right\rceil
\right]_+
\\
&\qquad\le
M
\sum_{k>
\left\lceil(1+\varepsilon)\log M/\log\log M\right\rceil}
\left(\frac{e\lambda}{k}\right)^k
\\
&\qquad=
o\left(\frac{\log M}{\log\log M}\right).
\end{align*}
For completeness, the ratio of consecutive summands is at most
\(2\lambda/k\) for all sufficiently large \(k\), while the logarithm of
the first summand is \(-\varepsilon\log M+o(\log M)\). It follows that
\[
\limsup_{M\to\infty}
\frac{\log\log M}{\log M}
\Ee\max_{1\le j\le M}N_j
\le1+\varepsilon.
\]
For the reverse inequality,
\[
\Ee\max_{1\le j\le M}N_j
\ge
(1-\varepsilon)\frac{\log M}{\log\log M}
P\left\{
\max_{1\le j\le M}N_j
\ge
(1-\varepsilon)\frac{\log M}{\log\log M}
\right\},
\]
and the probability on the right tends to one by Kimber's theorem. Letting
\(\varepsilon\downarrow0\) therefore proves
\(\Ee\max_{j\le M}N_j\sim\log M/\log\log M\), hence
\eqref{eq:taggedPoissonExactRate}.

We next verify that \(K_\lambda\) is finite using its Amemiya representation.
Fix \(0<t<\lambda\). Stirling's formula gives, as \(n\to\infty\),
\[
\log\left\{
e^{-\lambda}\frac{\lambda^n}{n!}
\exp\left\{h\!\left(\frac{tn}{\lambda}\right)\right\}
\right\}
=
-\left(1-\frac{t}{\lambda}\right)n\log n+O_{\lambda,t}(n).
\]
The resulting series is summable, so
\(\Ee_Q\psi_{\mathrm B}(tN/\lambda)<\infty\). Evaluating the infimum in
\eqref{eq:taggedPoissonAmemiyaNorm} at this value of \(t\) proves that
\(K_\lambda<\infty\).

Apply \Cref{lem:oneCoordinateAmemiya} with
\(\psi=\psi_{\mathrm B}\) to obtain the first two bounds in
\eqref{eq:taggedPoissonAmemiyaRate}. Finally,
\(h(u)\sim u\log u\), and therefore
\(\psi_{\mathrm B}^{-1}(M)\sim\log M/\log\log M\). This proves the last
comparison and completes the proof.
\end{proof}

We now compare the conclusion with more standard information-theoretic bounds. Let
\(P=\operatorname{Law}(1+N)\) and \(Q=\operatorname{Law}(N)\). Since
\(I(W;X)\le D_{\mathrm{KL}}(P\|Q)=:C_\lambda\), where
\(C_\lambda=\Ee\log\{(N+1)/\lambda\}<\infty\), classical Fano gives
\(S_{M,\lambda}^\star\le(C_\lambda+\log2)/\log M\); see
\cite[Chapter~7]{CoverThomas2006}. This proves that success vanishes, but it
exceeds the exact rare-success scale by a factor of order
\(M\log\log M/(\log M)^2\).

The power choice \(\psi_p(u)=u^p/p\), for fixed \(p>1\), gives
\(S_{M,\lambda}^\star\lesssim_{\lambda,p}M^{-1+1/p}\). The untailored
exponential choice \(\psi_{\mathrm E}(u)=e^u-u-1\) gives only
\(S_{M,\lambda}^\star\lesssim_\lambda(\log M)/M\). In contrast,
\(\psi_{\mathrm B}\) inserts the missing factor \(1/\log\log M\) and
recovers the exact order in \eqref{eq:taggedPoissonExactRate}.

There is no contradiction with the role of Maximal Leakage in exact
recovery. In fact, by the integral representation in
\eqref{eq:maximalLeakageIntegralDefinition},
\begin{equation}
\exp\{\ml{W}{X}\}
=
\Ee_{Q^{\otimes M}}\max_{j\le M}r(N_j),
\qquad
S_{M,\lambda}^\star
=
\frac{\exp\{\ml{W}{X}\}}{M}.
\label{eq:taggedPoissonLeakageIdentity}
\end{equation}
Thus Maximal Leakage is exact here. The contribution of the
Amemiya bound is computational rather than an abstract improvement
over Maximal Leakage: it bounds this integral at the correct scale using a
single norm matched to the likelihood ratio. Likewise, allowing the power
\(p=p_M\) to vary with \(M\) can reconstruct the same scale from the full
moment profile. The advantage of the Bennett-type Amemiya
formulation is that one fixed Young function packages that profile and
produces the appropriate multiple-testing penalty through
\(\psi_{\mathrm B}^{-1}(M)\).

Finally, under zero--one loss, the minimax quantile equals one precisely when
the minimax success probability is below \(1-\delta\). Hence
\eqref{eq:taggedPoissonAmemiyaRate} implies
\[
\cM(\delta)=\cM_-(\delta)=1
\quad\text{whenever}\quad
K_\lambda\frac{\psi_{\mathrm B}^{-1}(M)}{M}<1-\delta.
\]
Together with \eqref{eq:taggedPoissonExactRate}, this identifies the
high-confidence transition, up to constants, at
\(1-\delta\asymp_\lambda(\log M)/(M\log\log M)\). This example shows that
the Amemiya bound is useful when the success probability is controlled by
the tail of the likelihood ratio. A suitable Young function recovers the
correct scale even when fixed-order divergence bounds do not.
\section{Conclusion}

We developed a loss-adapted Neyman--Pearson metaconverse for minimax
quantiles. The central principle is that the small-ball function captures
how the prior concentrates in loss balls, while the inverse
Neyman--Pearson function controls the probability of successful
estimation. Relaxing this Neyman--Pearson bound in different ways yields
converses based on \(f\)-informativity, Sibson mutual information, Maximal
Leakage, and Amemiya norms, with classical Fano and Le Cam inequalities
appearing as special cases.

The applications show that different recovery criteria favour different
information measures. Maximal Leakage is naturally suited to exact recovery and
is exact for symmetric models admitting an equaliser rule. Finite-order Sibson
mutual information can be strictly stronger for approximate recovery, where
success permits a neighbourhood of the true parameter. Amemiya norms can
capture the effect of rare observations that fixed-order moments miss. The Gaussian weighted stochastic block
model and the bounded-energy low-rank denoising problem further
demonstrate that the framework produces explicit finite-sample minimax
bounds in both discrete and continuous settings.

Several questions remain open. It would be useful to identify general
conditions under which the Neyman--Pearson metaconverse is attained, or
under which a particular relaxation is sharp up to constants or
exponents. Another direction is to develop systematic procedures for
selecting the Sibson order or the Young function from the structure of
the success event and the likelihood ratio. A further extension would
condition the metaconverse on an auxiliary random variable, as in Xu and
Raginsky~\cite{XuRaginsky2017}. Finally, extending the
framework to interactive settings, sequential procedures, and
constrained decision rules may yield similarly unified
confidence-dependent converses in more general statistical problems.

\bibliographystyle{plainnat}
\bibliography{references}
\appendix
\section{Proofs}
\label{app:proofs}

\begin{proof}[Proof of~\Cref{thm:npLowerBound}]
Fix an estimator $\widehat \theta$, an auxiliary prior $P_W$, a channel $P_{X|W}$ and an arbitrary distribution over $\mathcal X$, $Q_X$. 
One has that
\begin{equation}
    P_{WX}(\ell(\widehat \theta(X),W)<\rho) = \mathbb E_{P_{WX}}\left[\1_{\left\{\ell(\widehat \theta(X),W)<\rho\right\}}\right].
\end{equation}
Moreover, one has that
\begin{align}
    E_{P_WQ_X}\left[\1_{\left\{\ell(\widehat \theta(X),W)<\rho\right\}}\right] &= \int P_{W}(\ell(\widehat \theta(x),W)<\rho) Q_X(\mathrm dx)\\
    &\leq \sup_{\vartheta \in \Theta} P_{W}(\ell(\vartheta,W)<\rho) \\
    &= L_W(\rho).
\end{align}
Consequently, one has that for every $Q_X$ and $P_W$:
\begin{equation}
    \beta_{P_{WX}(\ell(\widehat \theta(X),W)<\rho)}(P_{WX},P_WQ_X) \leq L_W(\rho).
\end{equation}
Thus, one can deduce that:
\begin{align}
    P_{WX}(\ell(\widehat \theta(X),W)<\rho) &\leq \sup\{ \alpha \in [0,1]: \beta_\alpha(P_{WX},P_WQ_X)\leq L_W(\rho)\}\\ &= \beta_{\circ}(P_{WX},P_WQ_X)^{-1}(L_W(\rho)).
\end{align}
Since this holds for every $Q_X$, one has:
\begin{equation}
    P_{WX}(\ell(\widehat \theta(X),W)<\rho) \leq \inf_{Q_X}\beta_{\circ}(P_{WX},P_WQ_X)^{-1}(L_W(\rho)).
\end{equation}
Hence, to conclude:
\begin{equation}
    P_{WX}(\ell(\widehat \theta(X),W)\geq \rho) \geq 1-\inf_{Q_X}\beta_{\circ}(P_{WX},P_WQ_X)^{-1}(L_W(\rho)).
\end{equation}
Since the minimax risk is always larger than the auxiliary Bayesian average $P_{WX}(\ell(\widehat \theta(X),W)\geq \rho)$ one can conclude that
\begin{equation}
\inf_{\widehat\theta}
\sup_{\theta}
\sup_{P_{X\mid\theta}\in\mathscr P_{X\mid\theta}}
P_{X\mid\theta}
\left(\ell(\widehat\theta(X),\theta)\ge\rho\right)
\ge
1-\inf_{Q_X}\beta_{\circ}(P_{WX},P_WQ_X)^{-1}(L_W(\rho)).
\end{equation}
\end{proof}

\begin{proof}[Proof of~\Cref{thm:fDivLowerBound}]
Set \(u:=L_W(\rho)\in[0,1]\) and
\(b:=\beta_{\circ}^{-1}(P_{WX},P_WQ_X)(u)\), and take
\(\alpha\in(u,b)\). By the definition of the generalised inverse,
\(\beta_\alpha(P_{WX},P_WQ_X)\leq u\). Take \(\epsilon>0\) such that
\(u+\epsilon<\alpha\). By the definition of \(\beta_\alpha\), there
exists a function \(\varphi_\epsilon\) such that
\begin{align}
    \mathbb{E}_{P_{WX}}[\varphi_\epsilon]&\geq \alpha\\
    \mathbb{E}_{P_WQ_X}[\varphi_\epsilon] &\leq \beta_\alpha(P_{WX},P_WQ_X)+\epsilon \leq u+\epsilon.
\end{align}
Consequently, one has that
\begin{equation}
     \mathbb{E}_{P_WQ_X}[\varphi_\epsilon] \leq u+\epsilon < \alpha \leq  \mathbb{E}_{P_{WX}}[\varphi_\epsilon].
\end{equation}
By the data-processing inequality for $f$-divergences one has that
\begin{equation}
    D_f(P_{WX}\|P_WQ_X) \geq d_f(\mathbb{E}_{P_{WX}}[\varphi_\epsilon]\|\mathbb{E}_{P_WQ_X}[\varphi_\epsilon]),
\end{equation}
where $d_f(p\|q)=qf(p/q)+(1-q)f((1-p)/(1-q))$ denotes the binary $f$-divergence. For a given $q$ the map $p \mapsto d_f(p\|q)$ is non-decreasing for $p\in[q,1]$. Similarly, for fixed $p$, the map $q\mapsto d_f(p\|q)$ is non-increasing for $q\in [0,p]$. 
One can thus show that \begin{align}
    d_f(\mathbb{E}_{P_{WX}}[\varphi_\epsilon]\|\mathbb{E}_{P_WQ_X}[\varphi_\epsilon]) &\geq d_f(\mathbb{E}_{P_{WX}}[\varphi_\epsilon]\| u+\epsilon)\\
    &\geq d_f(\alpha\|u+\epsilon).
\end{align}
Therefore, taking the limit $\epsilon\to 0$ one recovers the following for $\alpha \in (u, \beta_{\circ}^{-1}(P_{WX},P_WQ_X)(u)) $
\begin{equation}
    D_f(P_{WX}\|P_WQ_X) \geq d_f(\alpha\|u).
\end{equation}
Moreover, by lower semi-continuity of $d_f(\cdot\|u)$, letting $\alpha \uparrow \beta_{\circ}(P_{WX},P_WQ_X)^{-1}(u)$ yields
\begin{equation}
    D_f(P_{WX}\|P_WQ_X) \geq d_f(\beta_{\circ}(P_{WX},P_WQ_X)^{-1}(u)\|u). 
\end{equation}
Given that $d_f(\cdot \|u)$ is non-decreasing on $[u,1]$ we reach the following bound:
\begin{equation}
    \beta_{\circ}(P_{WX},P_WQ_X)^{-1}(u) \leq d_{f}^{-1}(D_f(P_{WX}\|P_WQ_X)\|u),
\end{equation}
where $d_f^{-1}(\cdot \|u)$ is the generalised inverse of
$p\mapsto d_f(p\|u)$ over $p\in[u,1]$.
Using this estimate in~\Cref{thm:npLowerBound} yields the result.
\end{proof}

\begin{proof}[Proof of~\Cref{cor:fFano}]
Since $P_W$ is uniform over
$\{\theta_1,\ldots,\theta_M\}$, for every $\vartheta\in\Theta$,
\begin{align}
    P_W\bigl(\ell(\vartheta,W)<\rho\bigr)
    &=
    \frac{1}{M}
    \left|
    \left\{
    j\in[M]:
    \ell(\vartheta,\theta_j)<\rho
    \right\}
    \right|\\
    &\leq \frac{1}{M}.
\end{align}
It follows that
\begin{equation}
    L_W(\rho)
    =
    \sup_{\vartheta\in\Theta}
    P_W\bigl(\ell(\vartheta,W)<\rho\bigr)
    \leq \frac{1}{M}.
\end{equation}
Applying~\Cref{thm:fDivLowerBound} with
$u=1/M$ yields the first claim. The second claim follows from the
$f$-informativity formulation in the preceding remark.
\end{proof}

\begin{proof}[Proof of~\Cref{cor:leCam}]
Let $W$ be uniformly distributed over $\{\theta_1,\theta_2\}$ and let
\[
    P_{X\mid W=\theta_i}=P_i,
    \qquad i\in\{1,2\}.
\]
By the separation assumption of the balls one has that for every $\vartheta\in\Theta$,
\begin{equation}
    P_W\left\{
    \ell(\vartheta,W)<\rho
    \right\}
    \leq\frac12,
\end{equation}
and therefore the following bound holds for the small-ball probability:
\begin{equation}
    L_W(\rho)\leq\frac12.
\end{equation}

Consider now the $f$-divergence generated by
\[
    f(t)=\frac12|t-1|,
\]
which corresponds to total variation. Its binary counterpart is
\begin{equation}
    d_f(p\|q)=|p-q|.
\end{equation}
Hence, for $p\in[q,1]$,
\begin{equation}
    d_f^{-1}(v\|q)
    =
    (q+v)\wedge1.
\end{equation}
Applying~\Cref{thm:fDivLowerBound} with $u=1/2$ yields
\begin{equation}
    P_{WX}
    \left\{
    \ell(\widehat\theta(X),W)\geq\rho
    \right\}
    \geq
    1-
    \left[
    \frac12+
    \inf_{Q_X}
    \operatorname{TV}(P_{WX},P_WQ_X)
    \right].
\end{equation}

Since $W$ is uniform over two points,
\begin{align}
    \operatorname{TV}(P_{WX},P_WQ_X)
    &=
    \frac12
    \left\{
    \operatorname{TV}(P_1,Q_X)
    +
    \operatorname{TV}(P_2,Q_X)
    \right\}.
\end{align}
By the triangle inequality,
\begin{equation}
    \operatorname{TV}(P_1,P_2)
    \leq
    \operatorname{TV}(P_1,Q_X)
    +
    \operatorname{TV}(P_2,Q_X),
\end{equation}
and therefore
\begin{equation}
    \inf_{Q_X}
    \operatorname{TV}(P_{WX},P_WQ_X)
    \geq
    \frac12\operatorname{TV}(P_1,P_2).
\end{equation}
On the other hand, choosing $Q_X=P_1$ gives
\begin{equation}
    \operatorname{TV}(P_{WX},P_WP_1)
    =
    \frac12\operatorname{TV}(P_1,P_2),
\end{equation}
so that
\begin{equation}
    \inf_{Q_X}
    \operatorname{TV}(P_{WX},P_WQ_X)
    =
    \frac12\operatorname{TV}(P_1,P_2).
\end{equation}
Consequently,
\begin{equation}
    P_{WX}
    \left\{
    \ell(\widehat\theta(X),W)\geq\rho
    \right\}
    \geq
    \frac12
    \left\{
    1-\operatorname{TV}(P_1,P_2)
    \right\}.
\end{equation}
The rest of the proof follows from similar arguments as above.
\end{proof}

\begin{proof}[Proof of~\Cref{cor:renyiSibson}]
By the data-processing inequality for R\'enyi divergence, for every $Q_X$,
\begin{equation}
\beta_{\circ}(P_{WX},P_WQ_X)^{-1}\bigl(L_W(\rho)\bigr)
\leq
d_\alpha^{-1}
\left(
D_\alpha(P_{WX}\|P_WQ_X)
\middle\|
L_W(\rho)
\right).
\end{equation}
Using this estimate in~\Cref{thm:npLowerBound} yields
\begin{equation}
P_{WX}\bigl(\ell(\widehat\theta(X),W)\geq\rho\bigr)
\geq
1-
\inf_{Q_X}
d_\alpha^{-1}
\left(
D_\alpha(P_{WX}\|P_WQ_X)
\middle\|
L_W(\rho)
\right).
\end{equation}

Now suppose that
\[
I_\alpha(W,X)
<
d_\alpha
\left(
1-\delta
\middle\|
L_W(\rho)
\right).
\]
By definition of Sibson mutual information, there exists a distribution
$\widetilde Q_X$ such that
\[
D_\alpha(P_{WX}\|P_W\widetilde Q_X)
<
d_\alpha
\left(
1-\delta
\middle\|
L_W(\rho)
\right).
\]
Since $p\mapsto d_\alpha(p\|L_W(\rho))$ is non-decreasing on
$[L_W(\rho),1]$, it follows that
\[
d_\alpha^{-1}
\left(
D_\alpha(P_{WX}\|P_W\widetilde Q_X)
\middle\|
L_W(\rho)
\right)
<
1-\delta.
\]
Hence
\[
\inf_{Q_X}
d_\alpha^{-1}
\left(
D_\alpha(P_{WX}\|P_WQ_X)
\middle\|
L_W(\rho)
\right)
<
1-\delta,
\]
and the conclusion follows.
\end{proof}

\begin{proof}[Proof of~\Cref{cor:amemiyaHolder}]
Let
\[
E_\rho
:=
\left\{
(x,w):
\ell(\widehat\theta(x),w)<\rho
\right\}.
\]
For an arbitrary auxiliary distribution $Q_X$, one has that
\begin{align}
P_{WX}(E_\rho)
&=
\mathbb E_{P_WQ_X}
\left[\frac{\mathrm dP_{WX}}{\mathrm d(P_WQ_X)}\1_{E_\rho}
\right] \\
&\overset{(a)}{\leq} \left\|
\frac{\mathrm dP_{WX}}{\mathrm d(P_WQ_X)}
\right\|_{L_\Phi^{\mathrm A}(P_WQ_X)}
\left\|
\1_{E_\rho}
\right\|_{L_\Psi^{\mathrm L}(P_WQ_X)} \\
&= \left\|
\frac{\mathrm dP_{WX}}{\mathrm d(P_WQ_X)}
\right\|_{L_\Phi^{\mathrm A}(P_WQ_X)} \frac{
1
}{
\Psi^{-1}
\left(
1/(P_WQ_X)(E_\rho)
\right)
} \\ 
&\overset{(b)}{\leq} \frac{
\left\|
\frac{\mathrm dP_{WX}}{\mathrm d(P_WQ_X)}
\right\|_{L_\Phi^{\mathrm A}(P_WQ_X)}
}{
\Psi^{-1}
\left(
1/L_W(\rho)
\right)
}.
\end{align}
where (a) follows from the generalised H\"older inequality, and (b) follows
from the monotonicity of $\Psi^{-1}$ together with
$(P_WQ_X)(E_\rho)\leq L_W(\rho)$.
Moreover, since this holds for every $Q_X$,
\begin{equation}
P_{WX}(E_\rho)
\leq
\frac{\inf_{Q_X} 
\left\|
\frac{\mathrm dP_{WX}}{\mathrm d(P_WQ_X)}
\right\|_{L_\Phi^{\mathrm A}(P_WQ_X)}
}{
\Psi^{-1}
\left(
1/L_W(\rho)
\right)
},
\end{equation}
which yields
\begin{equation}
P_{WX}\bigl(\ell(\widehat\theta(X),W)\geq\rho\bigr)
\geq
1-\frac{\inf_{Q_X}
\left\|
\frac{\mathrm dP_{WX}}{\mathrm d(P_WQ_X)}
\right\|_{L_\Phi^{\mathrm A}(P_WQ_X)}
}{
\Psi^{-1}
\left(
1/L_W(\rho)
\right)
}.
\end{equation}
The minimax-quantile conclusion follows from~\Cref{thm:npLowerBound}.
\end{proof}

\begin{proof}[Proof of~\Cref{rmk:holderLowerBound}]
The result follows from~\Cref{cor:amemiyaHolder} with
$\Phi(x)=x^\alpha/\alpha$, whose complementary Young function is
$\Psi(x)=x^{\alpha'}/\alpha'$, where
$1/\alpha+1/\alpha'=1$. The equivalent information form follows from the
variational representation of Sibson mutual information and the identity between
R\'enyi divergence and the $L^\alpha$-norm of the Radon--Nikodym derivative.
\end{proof}

\begin{proof}[Proof of~\Cref{thm:maximalLeakageExact}]
The loss-ball packing property implies that every estimator induces a
decoder on $\mathcal V$ whose exact-recovery success probability is at
least its success probability at threshold $\rho$. Since the MAP rule
maximises average exact-recovery success under the uniform prior, for
every estimator $\widehat\theta$,
\begin{align}
&\frac1M
\sum_{j=1}^M
P_{X\mid\theta_j}
\left(
\ell\bigl(\widehat\theta(X),\theta_j\bigr)<\rho
\right)
\notag\\
&\qquad\leq
P_{WX}
\left(
\widehat\theta_{\mathrm{MAP}}(X)=W
\right).
\label{eq:mapAverageOptimality}
\end{align}

By~\eqref{eq:mapEqualizer},
\begin{equation}
P_{WX}
\left(
\widehat\theta_{\mathrm{MAP}}(X)=W
\right)
=
s_{\mathrm{MAP}},
\end{equation}
and the worst-case success probability of the MAP rule is also
$s_{\mathrm{MAP}}$. Hence, for every $\widehat\theta$,
\begin{align}
&\min_{j=1,\ldots,M}
P_{X\mid\theta_j}
\left(
\ell\bigl(\widehat\theta(X),\theta_j\bigr)<\rho
\right)
\notag\\
&\qquad\leq
\frac1M
\sum_{j=1}^M
P_{X\mid\theta_j}
\left(
\ell\bigl(\widehat\theta(X),\theta_j\bigr)<\rho
\right)
\notag\\
&\qquad\leq s_{\mathrm{MAP}}.
\end{align}
The MAP rule attains the right-hand side, so
\begin{equation}
\sup_{\widehat\theta\in\mathcal E}
\min_{j=1,\ldots,M}
P_{X\mid\theta_j}
\left(
\ell\bigl(\widehat\theta(X),\theta_j\bigr)<\rho
\right)
=
s_{\mathrm{MAP}}.
\end{equation}
Taking complements proves~\eqref{eq:minimaxMapEquality}.

For a uniform prior, the operational characterisation of maximal
leakage gives
\begin{equation}
P_{WX}
\left(
\widehat\theta_{\mathrm{MAP}}(X)=W
\right)
=
\frac{\exp\{\ml{W}{X}\}}{M}.
\label{eq:mapLeakageIdentity}
\end{equation}
Moreover, by the Neyman--Pearson converse and its Maximal Leakage
relaxation,
\begin{align}
P_{WX}
\left(
\widehat\theta_{\mathrm{MAP}}(X)=W
\right)
&\leq
\inf_{Q_X}
\beta_{\circ}(P_{WX},P_WQ_X)^{-1}
\left(\frac1M\right)\\
&\leq
\frac{\exp\{\ml{W}{X}\}}{M}.
\end{align}
The first and last quantities coincide by~\eqref{eq:mapLeakageIdentity};
therefore both inequalities are equalities. This proves
\eqref{eq:mapNPLeakageEquality} and completes the proof.
\end{proof}
\end{document}